%% file: main.tex
\documentclass[11pt]{article}
\usepackage{booktabs} 
\usepackage[ruled,linesnumbered]{algorithm2e} 

\SetAlFnt{\small}
\SetAlCapFnt{\small}
\SetAlCapNameFnt{\small}
\SetAlCapHSkip{0pt}
\IncMargin{-\parindent}

\usepackage{fullpage}
\usepackage{amsmath, amsthm, amssymb}
\usepackage{verbatim}
\usepackage[sort]{natbib}
\usepackage{xcolor}
\usepackage{graphicx}
\usepackage{hyperref}

\newtheorem{lemma}{Lemma}
\newtheorem{theorem}{Theorem}
\newtheorem{corollary}{Corollary}

\theoremstyle{definition}

\newcommand{\bE}{\operatorname*{\mathbb{E}}}
\newcommand{\bR}{\mathbb{R}}

\newcommand{\eps}{\varepsilon}
\newcommand{\cA}{\mathcal{A}}
\newcommand{\cS}{\mathcal{S}}
\newcommand{\cH}{\mathcal{H}}
\newcommand{\cD}{\mathcal{D}}
\newcommand{\trans}{P}
\newcommand{\si}{{s_\mathrm{init}}}
\newcommand{\st}{{s_\mathrm{term}}}
\newcommand{\pf}{f}
\newcommand{\fr}{F}
\newcommand{\argmax}{\operatorname*{argmax}}

\newcommand{\ap}{\mathsf{ap}}
\newcommand{\pa}{\mathsf{pa}}
\newcommand{\pp}{\mathsf{pp}}
\newcommand{\eval}{\mathtt{eval}}
\newcommand{\opt}{\mathrm{opt}}
\newcommand{\error}{-1}
\newcommand{\obj}{\mathrm{obj}}

\renewcommand{\max}{\operatorname{max}}

\author{
    Hanrui Zhang \\
    Chinese University of Hong Kong \\
    \texttt{hanrui@cse.cuhk.edu.hk}
    \and
    Yu Cheng \\
    Brown University \\
    \texttt{yu\_cheng@brown.edu}
    \and
    Vincent Conitzer \\
    Carnegie Mellon University \\
    \texttt{conitzer@cs.cmu.edu}
}
\date{}

\title{Efficiently Solving Turn-Taking Stochastic Games with Extensive-Form Correlation}

\begin{document}

\maketitle

\begin{abstract}
We study equilibrium computation with extensive-form correlation in two-player turn-taking stochastic games.
Our main results are two-fold: (1) We give an algorithm for computing a Stackelberg extensive-form correlated equilibrium (SEFCE), which runs in time polynomial in the size of the game, as well as the number of bits required to encode each input number.
(2) We give an efficient algorithm for approximately computing an optimal extensive-form correlated equilibrium (EFCE) up to machine precision, i.e., the algorithm achieves approximation error $\eps$ in time polynomial in the size of the game, as well as $\log(1 / \eps)$.

Our algorithm for SEFCE is the first polynomial-time algorithm for equilibrium computation with commitment in such a general class of stochastic games.
Existing algorithms for SEFCE typically make stronger assumptions such as no chance moves, and are designed for extensive-form games in the less succinct tree form.
Our algorithm for approximately optimal EFCE is, to our knowledge, the first algorithm that achieves 3 desiderata simultaneously: approximate optimality, polylogarithmic dependency on the approximation error, and compatibility with stochastic games in the more succinct graph form.
Existing algorithms achieve at most 2 of these desiderata, often also relying on additional technical assumptions.
\end{abstract}

\section{Introduction}

Equilibrium computation is one of the most important topics in algorithmic game theory.
Decades of effort has painted a fairly complete landscape for the computational complexity of various equilibrium concepts in normal-form 
games: Roughly speaking, computing an (optimal) equilibrium is computationally tractable, if either correlation is allowed between both players' actions, or one of the two players has the power to commit to a strategy.
In other words, in normal-form games, there are polynomial-time algorithms for computing an optimal (i.e., maximizing a convex combination of the two players' utilities) correlated equilibrium, and for computing a Stackelberg equilibrium (see, e.g., \citep{papadimitriou2007complexity}).

The situation is subtler for games in dynamic environments, where the two players iteratively take actions, each affecting the state of the world, and together determining the overall payoff of each player.
Such games are conventionally modeled as stochastic games or extensive-form games (we will discuss the differences between the two formulations momentarily).
In these games, neither correlation nor commitment power guarantees tractability.
In fact, it is $\mathsf{NP}$-hard to compute a Stackelberg equilibrium in these games, even if normal-form correlation --- meaning a mediator recommends a whole strategy, consisting of an action to be played in each possible information set of the game, to each player at the beginning of the play --- is allowed \citep{letchford2010computing}.
In light of such hardness results, it has been argued that the right notion of correlation in extensive-form games is extensive-form correlation, where the mediator reveals a recommended action to be played in an information set only when a player has reached that information set.

The notion of extensive-form correlation leads to a number of natural solution concepts, which generalize correlated equilibria in normal-form games and are computationally more tractable.
Among them, the most important and well-studied ones are extensive-form correlated equilibria (henceforth EFCE) and Stackelberg extensive-form correlated equilibria (henceforth SEFCE) \citep{von2008extensive}.
While significant effort has been made on designing efficient algorithms for computing (optimal) EFCE and SEFCE, most existing algorithms are designed for extensive-form games in the tree form (for exceptions, see Section~\ref{sec:related_work}): The input to such an algorithm is by default a tree capturing all possibilities in a game, where each leaf corresponds to a possible way for the game to play out, and the time complexity of the algorithm is polynomial in the size of this game tree.
Such a game tree is often not the most succinct representation of a game.
For example, consider the following adapted version of the game Nim \citep{bouton1901nim}: Initially there are $k$ matches on the table.
Alice and Bob take turns removing matches, where in each turn, the acting player can remove either $1$ or $2$ matches.
The player who removes the last match wins.
The natural state space of the game is quite succinct: The state is fully determined by the number of matches left and the identity of the player to act next, so the state space is of size $O(k)$.
However, the tree form of the game must encode the entire history through which a state is reached (e.g., ``Alice removes 1 match; Bob removes 2 matches; Alice removes 2 matches; Bob removes 1 match; \dots''), which means the game tree has $2^{\Omega(k)}$ nodes.
In such cases, an algorithm that runs in time polynomial in the size of the game tree would not appear particularly efficient.
In this paper, we address the above issue by designing efficient algorithms for optimal EFCE and SEFCE that work with stochastic games, which are by default represented in the graph form that succinctly encodes a game.

\subsection{Our Results}

Throughout the paper, we focus on two-player, finite-horizon, turn-taking stochastic games.
Put in different words, we focus on two-player, perfect-information extensive-form games in the graph form.
Our main results are twofold:
\begin{itemize}
    \item We give an algorithm for computing an SEFCE, which runs in time polynomial in the size of the game, as well as the number of bits required to encode each input number.
    \item We give an efficient algorithm for computing an approximately optimal approximate EFCE up to machine precision, i.e., the algorithm achieves approximation error $\eps$ in time polynomial in the size of the game, as well as $\log(1 / \eps)$.
\end{itemize}

Our algorithm for SEFCE is, to our knowledge, the first polynomial-time algorithm for equilibrium computation with commitment in such a general class of stochastic games (the main assumption being that the game is turn-taking).
As discussed in Section~\ref{sec:related_work}, existing algorithms for SEFCE typically make stronger assumptions such as no chance moves, and are designed for extensive-form games in the less succinct tree form.
Our algorithm for approximately optimal EFCE is, to our knowledge, the first algorithm that achieves 3 desiderata simultaneously: approximate optimality, polylogarithmic dependency on the approximation error, and compatibility with stochastic games in the graph form.
As discussed in Section~\ref{sec:related_work}, existing algorithms typically achieve at most 2 of these desiderata, often also relying on additional technical assumptions.

Technically, our algorithms are built on ideas fundamentally different from the most commonly seen techniques in equilibrium computation, i.e., linear programming and no-regret dynamics.
We take a semi-combinatorial approach centered around the notion of Pareto frontier curves.
Roughly speaking, the Pareto frontier curve for a state-action pair captures the optimal tradeoff between the two players' onward utilities subject to equilibirum conditions, in the subgame induced by the state-action pair.
These curves can be viewed as a multidimensional generalization of the $Q$-function commonly used in planning and reinforcement learning.
Given the right equilibrium conditions, computing an SEFCE or an optimal EFCE reduces to evaluating Pareto frontier curves, which at first sight appears to be a numerical problem in nature.
In order to perform the necessary evaluations efficiently, we establish combinatorial properties of the Pareto frontier curves, including recursive relations between curves for different state-action pairs, as well as lower bounds on the numerical ``resolution'' of the curves (i.e., how close two turning points can be on a curve).
The curves for SEFCE exhibit particularly nice properties, based on which we are able to design an essentially combinatorial procedure for evaluating the Pareto frontier curves recursively.
This involves binary searching over ``directions of evaluation'' up to a carefully chosen precision, as well as a memorization technique that avoids redundant recursive calls.
For EFCE, the Pareto frontier curves are less structured, and in particular, the curves can be very (i.e., doubly exponentially) fine in terms of their numerical resolution.
It thus becomes infeasible to exactly evaluate the curves, and our algorithm instead performs evaluations up to any desired precision in polylogarithmic time.
For a more detailed overview of our algorithms, see Section~\ref{sec:overview}.

\subsection{Further Related Work}
\label{sec:related_work}

Equilibrium computation in normal-form games has been extremely well-studied.
For example, without commitment, \citet{daskalakis2009complexity} and \citet{chen2009settling} show that computing a Nash equilibrium is $\mathsf{PPAD}$-complete in two-player normal-form games, and computing optimal Nash equilibria is generally 
$\mathsf{NP}$-hard~\citep{Gilboa89:Nash,Conitzer03:Nash}.
In contrast, when correlation is allowed, one can compute an optimal correlated equilibrium efficiently (see, e.g., \citep{papadimitriou2007complexity}).
With commitment, \citet{conitzer2006computing} give an efficient algorithm for computing a Stackelberg equilibrium in two-player normal-form games (see also \citet{Stengel10:Leadership}), and that this becomes hard with 3 players; however, 
if the committing player can also send signals to the other players, thereby effectively taking over the role of the mediator in correlated equilibrium, then the problem is again efficiently solvable with any number of players~\citep{Conitzer11:Commitment}.
So in short, efficient equilibrium computation is possible if either correlation or commitment is allowed.

Equilibrium computation becomes more difficult in dynamic environments, such as extensive-form games and stochastic games.
There, commitment does not imply efficient computation anymore: \citet{letchford2010computing} show that it is $\mathsf{NP}$-hard to compute a Stackelberg equilibrium in two-player extensive-form games, even with perfect information.
Moreover, their hardness result holds even if normal-form correlation (as opposed to extensive-form correlation to be discussed momentarily) is allowed.
Similar hardness results hold for various structured families of stochastic games \citep{letchford2012computing}.
To circumvent such hardness results, \citet{von2008extensive} introduce the notion of extensive-form correlation, where conceptually, recommended actions are revealed on the fly.
They give an efficient algorithm for computing an optimal extensive-form correlated equilibrium (EFCE) when the game has no chance moves, and show that with chance moves, the same problem is $\mathsf{NP}$-hard.
Notably, their hardness result assumes imperfect information, which turn-taking stochastic games do not have.
In short, extensive-form correlation is generally necessary for efficient computation in dynamic environments.

Subsequently, there has been a long line of research on the computation of optimal EFCE, as well as its Stackelberg version, Stackelberg EFCE (henceforth SEFCE).
However, as far as we know, all existing positive results, i.e., efficient algorithms, are for extensive-form games in the tree form.
To name a few examples, \citet{cermak2016using} give a polynomial-time algorithm for computing an SEFCE in extensive-form games without chance moves.
Relatedly, \citet{bovsansky2017computation} show, among other results, that it is possible to compute an SEFCE in perfect-information extensive-form games in polynomial time.
\citet{farina2019correlation} give a saddle-point formulation for optimal EFCE and design gradient-based algorithms that scale better in practice.
\citet{farina2020polynomial} give a polynomial-time algorithm for computing optimal EFCE when chance moves are public.
\citet{zhang2022optimal} give fixed-parameter algorithms for computing optimal EFCE, as well as related solution concepts.
\citet{zhang2022polynomial} give a polynomial-time algorithm for a general class of equilibrium computation problems that involve a mediator, which in particular generalize the results by \citet{zhang2022optimal}.
The type of computational problem considered in these results is ``easier'' than ours, in the sense that the graph form that we consider can be much more succinct (and never less succinct) than the tree form of the same game, and in such cases, an algorithm that runs in time that is polynomial in the size of the graph form is much more efficient.
Conversely, an algorithm that runs in time that is polynomial in the size of the tree form is not necessarily polynomial-time in the graph form.

Another line of research studies no-regret dynamics that converge to an arbitrary (i.e., not necessarily optimal) EFCE.
This translates to polynomial-time algorithms in $1 / \eps$ for $\eps$-EFCE.
Compared to these algorithms, our algorithm for approximately optimal EFCE (1) guarantees approximate optimality, (2) runs in polynomial time in $\log(1 / \eps)$ instead of $1 / \eps$, and (3) works with the more succinct graph form of the game, instead of the tree form.
\citet{huang2008computing} show how to compute an arbitrary EFCE exactly in polynomial time.
However, their algorithm does not guarantee optimality, nor is it compatible with the graph form.

Technically, computing optimal or Stackelberg equilibria generalizes the problem of planning in Markov decision processes under constraints.
Particularly related to our results is ``planning with participation constraints'' \citep{zhang2022planning,zhang2022efficient}: Roughly speaking, in that problem, the principal (corresponding to the leader in a Stackelberg game) chooses which action to take in each state, subject to the constraint that the agent (corresponding to the follower) is always willing to participate, i.e., the agent's onward utility in each state is always nonnegative.
This can be viewed as a highly restricted class of turn-taking stochastic games, where the agent's (follower's) only actions in each state are to stay and to quit.
\citet{zhang2022planning,zhang2022efficient} show that even in these restricted environments, an optimal policy (i.e., a Stackelberg equilibrium) may have to be history-dependent, and give a polynomial-time algorithm for planning with participation constraints.
However, their algorithm is tailored to the essentially non-strategic setting where the agent's power is extremely limited.
In contrast, we consider general game-theoretic settings where the two players are generally equally powered, except that in the Stackelberg setting, one player in addition has commitment power.

Pareto frontier curves have also been (implicitly) considered in prior work in equilibrium computation, e.g., \citep{letchford2010computing,bovsansky2017computation}.
In particular, \citet{letchford2010computing} and \citet{bovsansky2017computation} present algorithms that essentially keep track of the entire Pareto frontier curves, by computing and storing all turning points on each curve.
In the settings that we study, doing so would generally require exponential computation --- in fact, the key idea behind our algorithms is to avoid computing and storing all turning points, and only focus on the important ones.

\section{Preliminaries}

\paragraph{Stochastic games.}
We focus on two-player finite-horizon turn-taking stochastic games in this paper.
There is a finite set of states $\cS = [n]$, and a finite set of actions $\cA = [m]$.
$\si = 1$ and $\st = n$ are the initial and terminal states, respectively.
For each state $s \in \cS$, there is an acting player $\ap(s) \in \{1, 2\}$, who unilaterally decides which action to play in state $s$.
For each player $i \in \{1, 2\}$, there is a reward function $r_i: \cS \times \cA \to \bR_+$, which specifies the immediate reward $r_i(s, a)$ that player $i$ receives when action $a$ is played in state $s$.
We assume rewards are normalized, i.e., $r_i(s, a) \in [0, 1]$ for each $i \in \{1, 2\}$, $s \in \cS$ and $a \in \cA$.
A transition operator $\trans: \cS \times \cA \to \Delta(\cS)$ specifies the distribution $\trans(s, a)$ of the next state when action $a$ is played in state $s$, where for each $s'$, $\trans(s, a, s')$ is the probability that the next state is $s'$.

Unless otherwise specified, we assume the transition operator is acyclic, i.e., $\trans(s, a, s') > 0$ only if $s' > s$ or $s = s' = n$.
For the terminal state $\st = n$ in particular, we assume $r_i(n, a) = 0$ and $\trans(n, a, n) = 1$ for each action $a \in \cA$.
In other words, there is no meaningful action in the terminal state $\st = n$.
These assumptions essentially mean the game is finite-horizon.
In particular, note that in the finite-horizon case, the acyclicity assumption is without loss of generality, as the state could include the index of the current period (with a blowup proportional to the time horizon $T$).

\paragraph{Histories, strategies, and utilities.}
Fix a stochastic game $(\cS, \cA, \ap, r_1, r_2, \trans)$.
A history $h$ of length $t$ is a sequence of $t$ states and $t$ actions $h = (s_1, a_1, s_2, a_2, \dots, s_t, a_t)$ which fully describes $t$ consecutive steps of a play.
Let $\cH$ be the collection of all histories of all lengths not exceeding $n$ (so $\cH$ is finite).
For brevity, we also let $|h|$ denote the length of $h$, and $h + (s, a)$ be the history obtained by appending $(s, a)$ to the end of $h$.

A deterministic (history-dependent) strategy $\pi: \cH \times \cS \to \cA$ maps each history-state pair $(h, s)$ to the action to be played in $s$ given history $h$.
Note that we do not explicitly partition a strategy into two parts corresponding to the two players, since such a partition is induced by the mapping $\ap$ from each state to the corresponding acting player.
A randomized strategy $\Pi$ is a distribution over deterministic strategies.
For any deterministic strategy $\pi$, we say a history $h = (s_1, a_1, \dots, s_t, a_t)$ is admissible if the action played in each step is the one specified by $\pi$, i.e., for each $t' \in [t]$, $\pi((s_1, a_1, \dots, s_{t' - 1}, a_{t' - 1}), s_{t'}) = a_{t'}$.
For any randomized strategy $\Pi$, we say a history $h$ is admissible under $\Pi$ if $h$ is admissible under some $\pi$ in the support of $\Pi$.
Let $\cH^\pi$ (resp.\ $\cH^\Pi$) be the set of admissible histories under $\pi$ (resp.\ $\Pi$).
For a randomized strategy $\Pi$ and an admissible history $h \in \cH^\Pi$ under $\Pi$, let $\Pi \mid h$ denote the conditional version of $\Pi$ given that the states reached and actions played in the first $|h|$ steps are $h$. 

For each $i \in \{1, 2\}$, the onward utility $u_i^\pi(h, s)$ of a player $i$, under a deterministic strategy $\pi$, in state $s$, given history $h$, is
\[
    u_i^\pi(h, s) = \bE_{\{s_t\}_t}\left[\sum_{t \in [n]} r_i(s_t, a_t)\right],
\]
where $s_1 = s$, $h_1 = h$, $a_t = \pi(h_t, s_t)$ for each $t \in [n]$, $h_t = h_{t - 1} + (s_{t - 1}, a_{t - 1})$ for each $t \in \{2, \dots, n\}$, and $s_t \sim \trans(s_{t - 1}, a_{t - 1})$ for each $t \in \{2, \dots, n\}$.
Here we only need to consider $n$ steps, since the maximum meaningful length of a play is $n$ --- after $n$ steps, we must have reached the state $n$, and any extra steps would give reward $0$ to both players.
The unconditional onward utility $v_i^\Pi(h, s)$ of player $i$, under a randomized strategy $\Pi$, in state $s$, given history $h$, is
\[
    v_i^\Pi(h, s) = \bE_{\pi \sim \Pi}[u_i^\pi(h, s)].
\]
Note that this unconditional onward utility is not the actual expected onward utility; we need it mostly for notational simplicity.
Conceptually, it is the utility that player $i$ expects to receive in the future if they believe the posterior strategy is $\Pi$.
Given the unconditional onward utility, for each $i \in \{1, 2\}$, the actual onward utility $u_i^\Pi(h, s)$ of player $i$, under strategy $\Pi$, in state $s$, given history $h$, is simply
\[
    u_i^\Pi(h, s) = v_i^{\Pi \mid h}(h, s).
\]

\paragraph{Extensive-form correlated equilibria.}
Extensive-form correlated equilibria (EFCE) and their Stackelberg version (Stackelberg EFCE, or SEFCE) are the natural generalizations of correlated equilibria and Stackelberg equilibria to dynamic settings such as stochastic games and extensive-form games.
In the original definition of EFCE for extensive-form games by \citet{von2008extensive}, a mediator specifies a distribution over deterministic strategies (i.e., a randomized strategy according to our definition above), where each deterministic strategy specifies a recommended action in each node of the game tree (corresponding to a history-state pair in our formulation).
A deterministic strategy is drawn and fixed at the beginning of the play, but the recommended action in each node given by this strategy is revealed to the acting player only when the node is actually reached.
If a player decides to not follow a recommended action, that player will not receive recommended actions in the rest of the play.

For any $\eps \ge 0$, we say a player is $\eps$-best responding under a randomized strategy if that player cannot increase their onward utility by more than $\eps$ by deviating from the recommended action at any point of a recommended path of play, i.e., an admissible history.
A randomized strategy is an $\eps$-EFCE if both players are $\eps$-best responding. 
Moreover, consider a Stackelberg setting where player $1$ is the leader and player $2$ is the follower.
Then, a randomized strategy is an SEFCE if player $1$'s utility is maximized subject to the constraint that player $2$ is $0$-best responding (or simply best responding).

Put in our language, a player $i \in \{1, 2\}$ is $\eps$-best responding under a randomized strategy $\Pi$ iff for any admissible history $h \in \cH^\Pi$, state $s$ where $\ap(s) = i$, action $a$ where $h + (s, a) \in \cH^\Pi$, and deterministic strategy $\pi'$ where $\pi'(h, s) \ne a$:
\[
    v_i^{\Pi \mid (h + (s, a))}(h, s) \ge \eps + \bE_{\pi \sim \Pi \mid (h + (s, a))}\left[u_i^{(i: \pi', 3 - i: \pi)}(h, s)\right].
\]
Here, $(i: \pi', 3 - i: \pi)$ denotes a strategy obtained by combining $\pi'$ restricted to player $i$'s actions and $\pi$ restricted to $(3 - i)$'s actions (note that $3 - i = 2$ when $i = 1$, and vice versa; $3 - i$ simply means the other player than $i$).
That is,
\[
    (i: \pi', 3 - i: \pi)(h, s) = \begin{cases}
        \pi'(h, s), & \text{if } \ap(s) = i \\
        \pi(h, s), & \text{otherwise}.
    \end{cases}
\]
We will use this notation repeatedly in the rest of the paper.
The left hand side is the utility $i$ expects to receive if both players keep following the recommendations, where in particular, $i$'s belief for the strategy is $\Pi \mid (h + (s, a))$ because $i$ has already received the recommended action $a$.
The right hand side is the utility $i$ expects to receive if $i$ unilaterally deviates to $\pi'$ and the other player keeps following the recommendations, which should never be larger than the left hand side.
Recall that a randomized strategy $\Pi$ is an $\eps$-EFCE iff both players are $\eps$-best responding under $\Pi$.
One can check this is in fact equivalent to the definition by \citet{von2008extensive} when $\eps = 0$.
A randomized strategy $\Pi$ is an SEFCE (with player $1$ being the leader) iff
\[
    \Pi \in \argmax_{\text{player } 2 \text{ is best responding under } \Pi} u_1^\Pi(\emptyset, \si).
\]

\section{Stackelberg Extensive-Form Correlated Equilibria}
\label{sec:sefce}

\subsection{Overview of Our Approach}
\label{sec:overview}

\paragraph{Maximum punishment without loss of generality.}
Our algorithm is based on the standard observation that in an SEFCE, it is without loss of generality to maximally punish the follower when they deviate from the prescribed path of play, regardless of how that would affect the leader's utility at that point.
In fact, for any SEFCE, there exists an effectively equivalent SEFCE where deviation always immediately triggers maximum punishment, so once the follower deviates, the game immediately becomes effectively zero-sum.
This is because intuitively, the sole purpose of the leader's equilibrium strategy in parts of the game where the follower has deviated is to threaten the follower and cancel out any potential incentive to deviate.
In particular, such a threat would never be actually executed, because in equilibrium no player would deviate in the first place.
As such, it never hurts to threaten with the worst punishment possible.
This greatly simplifies the problem from a computational perspective, since computing a strategy for maximum punishment is no harder than solving turn-taking zero-sum stochastic games, which can be done by simple backward induction.

\paragraph{Reducing to constrained planning.}
Once the punishment strategy is fixed, we only need to optimize over strategies where the follower never faces worse utility than what they would face after deviating in the optimal way and being maximally punished thereafter.
One key observation here is that in any state, regardless of the recommended action, the optimal way to deviate is always the same.
So, to prevent the follower from deviating, we only need to guarantee that conditioned on the recommended action, the onward utility of the follower is at least the utility resulting from deviating optimally.
In particular, the latter utility depends only on the state (which is only true in turn-taking stochastic games).
Given this observation, the problem becomes a constrained planning problem, where we want to find an optimal strategy subject to the constraint that in each state where the follower is the acting player, the onward utility of the follower (conditioned on the recommended action) is at least some state-dependent quantity that can be efficiently pre-computed.
This is very similar to planning in constrained MDPs, except for one key difference: In constrained MDPs, typically there are a constant number of feasibility constraints over the cumulative reward vector, whereas in our constrained planning problem, there are separate feasibility constraints on the follower's onward utility in each state of the game where the follower is the acting player.

\paragraph{Pareto frontier curves and pivotal points.}
The way we approach the constrained planning problem is by considering the Pareto frontier curves for each state-action pair.
Roughly speaking, the Pareto frontier curve $\pf_{s, a}$ for a state-action pair $(s, a)$ captures the Pareto-optimal way to trade off between the two players' onward utilities after the acting player takes action $a$ in state $s$, subject to feasibility constraints {\em in the future}.
This can be viewed as a generalization of the $Q$-function that is commonly considered in reinforcement learning that captures the tradeoff between the two players' utilities.
For technical reasons, we intentionally disregard the constraint (if there is one) in the (current) state $s$.
With this definition, the problem of constrained planning (or at least, the problem of computing the maximum objective value therein) becomes the problem of evaluating the Pareto frontier curves subject to feasibility constraints.
In particular, as illustrated in Figure~\ref{fig:overview}, given an objective direction (which can be any combination of the two player's utilities), the maximum objective value onward for each state-action pair is simply the farthest point on (the feasible part of) the Pareto frontier curve along that direction.

\begin{figure}[t]
    \centering
    \includegraphics[width=0.48\linewidth]{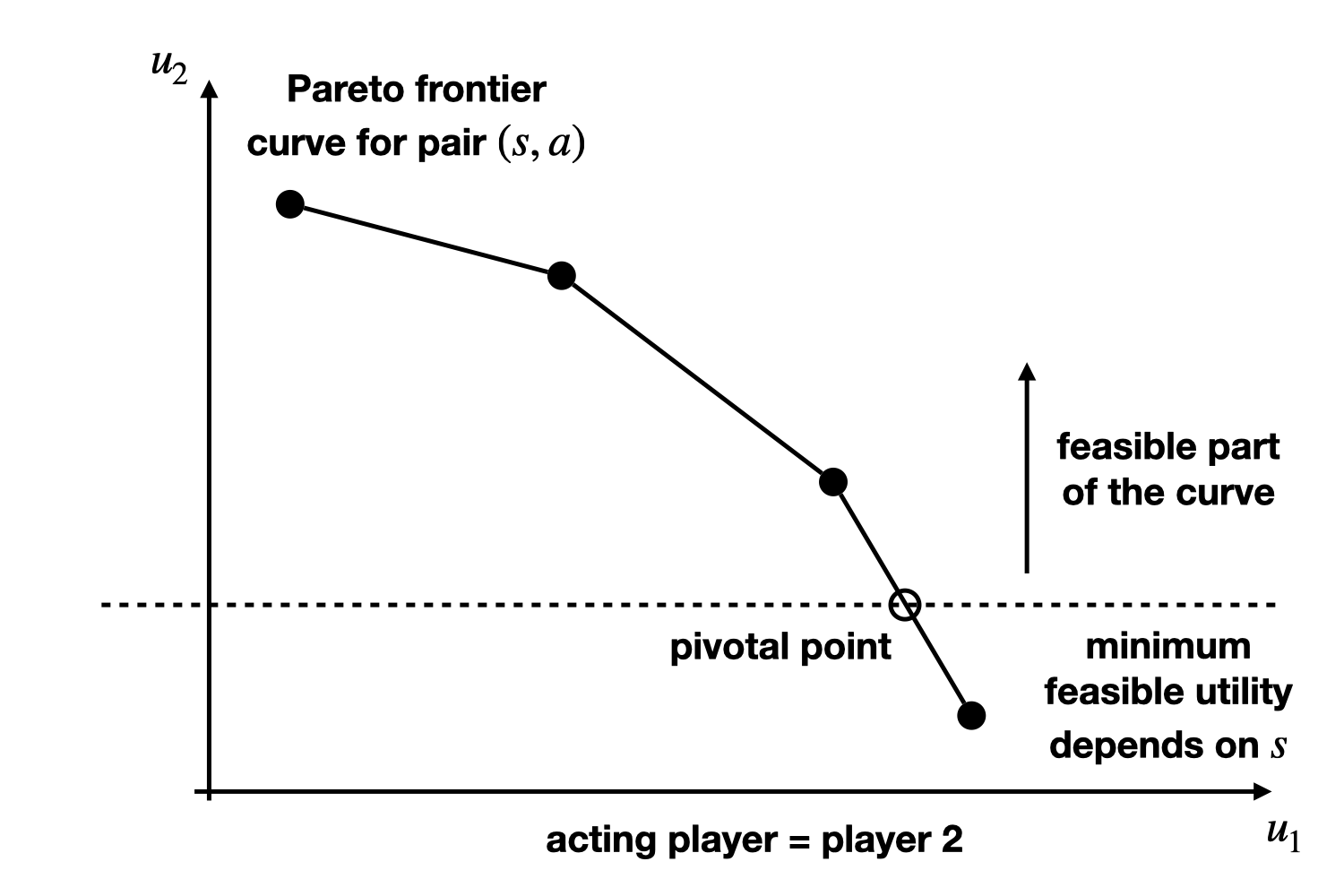}
    \includegraphics[width=0.48\linewidth]{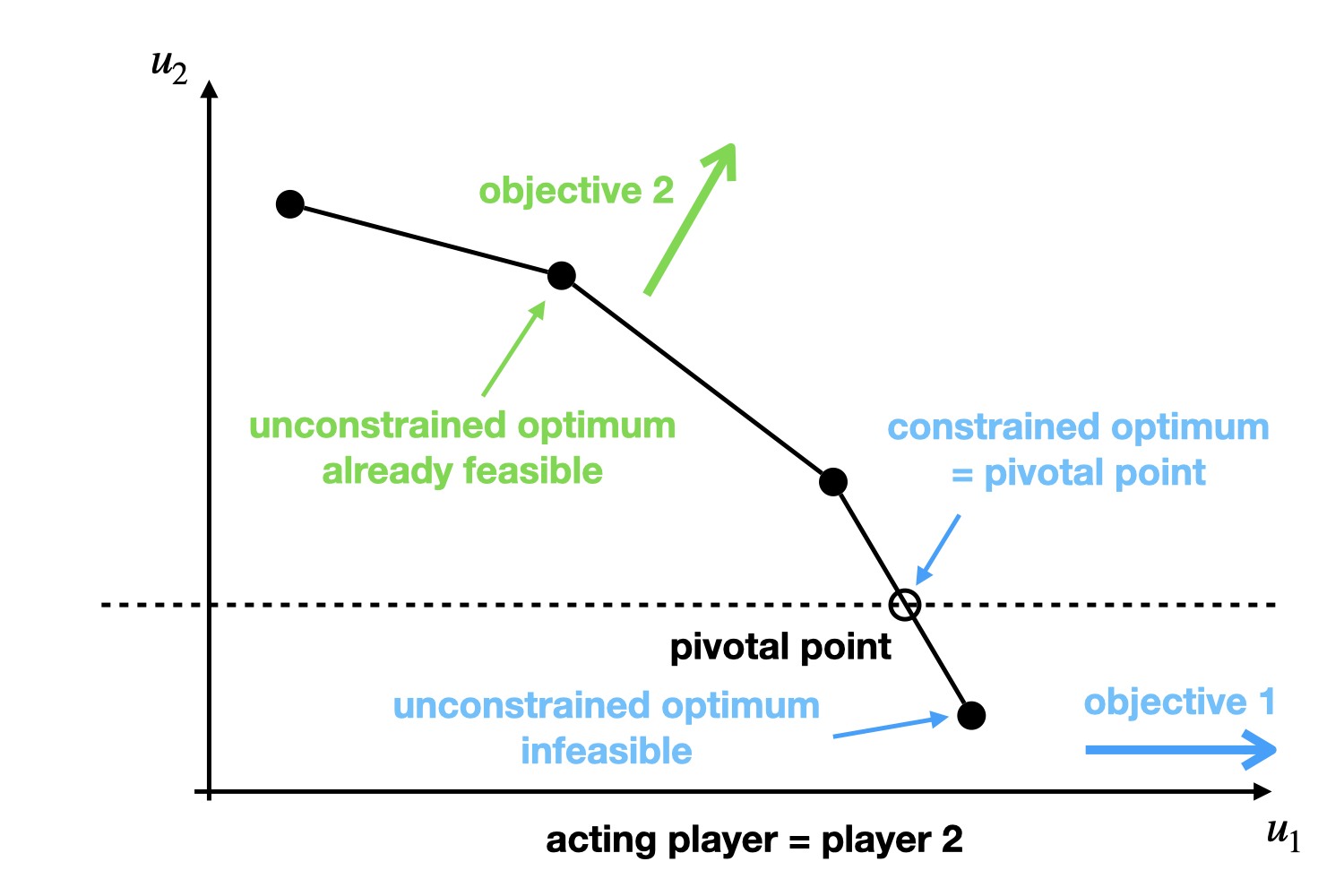}
    \caption{Illustration of Pareto frontier curves.}
    \label{fig:overview}
\end{figure}

Observe that under our definition, Pareto frontier curves are always concave, since from each state-action pair onward, the convex combination of two feasible strategies is also feasible --- this is true because we consider extensive-form correlation, so when we take a convex combination through randomization, players are not allowed to know the realization before reaching the state where the randomization happens.
Given this observation, finding this point involves two conceptual steps: One first finds the unconstrained optimum on the curve.
If that point turns out to be infeasible, then one ``rounds'' that point to the nearest feasible point, which is always the pivotal point (hollow points in Figure~\ref{fig:overview}), i.e., the unique point on the curve where the feasibility constraint is binding.
This highlights the importance of pivotal points, which play a central role in our algorithm.

\paragraph{Evaluating Pareto frontier curves.}
Now the problem becomes efficiently evaluating Pareto frontier curves subject to feasibility constraints.
At a high level, this can be done in a recursive fashion: Suppose we want to evaluate the curve $\pf_{s, a}$ for $(s, a)$ in a certain direction $\alpha \in \bR^2$, subject to the feasibility constraint in state $s$.
Moreover, suppose we can efficiently evaluate the constrained curves for all later state-action pairs.
Then we can perform the evaluation using the following procedure:
\begin{enumerate}
    \item For each later state $s' > s$:
    \begin{enumerate}
        \item For each action $a'$, evaluate the constrained curve $\pf_{s', a'}$ along direction $\alpha$.
        \item Let the farthest point along $\alpha$ found in the above evaluation be the partial result for state $s'$.
    \end{enumerate}
    \item Aggregate the partial results for all later states $s'$ according to the transition probabilities $\trans(s, a, s')$, and shift the aggregated result by the immediate rewards $(r_1(s, a), r_2(s, a))$ induced by $(s, a)$.
    \item If the shifted result above is feasible, return it; otherwise, return the pivotal point on $\pf_{s, a}$.
\end{enumerate}

There is one gap in the above procedure: In step 3, when we ``round'' an infeasible shifted result, it is assumed that we already know the pivotal point on the curve $\pf_{s, a}$ --- in fact, without this rounding step, we would be evaluating the curve $\pf_{s, a}$ without the feasibility constraint in $s$.
In reality we need to compute this pivotal point efficiently.
In what follows we discuss how this can be done up to machine precision.

Again assume we can evaluate the curves for all later state-action pairs.
To approximate the pivotal point, we only need to find two points on $\pf_{s, a}$ that are close enough, to the left and the right of the pivotal point respectively.
Then a particular convex combination of the two points will be a good approximation of the pivotal point.
Suppose we want to find a point to the left of the pivotal point that is close enough, and for concreteness, suppose the acting player is player 2, as in the left subfigure in Figure~\ref{fig:overview}.
There, a point is to the left of the pivotal point if and only if it is feasible.
Conceptually, the pivotal point can be found using a ``moving direction'' procedure: We start with direction $(0, 1)$, and perform an evaluation of $\pf_{s, a}$ in that direction without the feasibility constraint in $s$.
Such an unconstrained evaluation can be done recursively without knowing the pivotal point (using the procedure above without step 3). 
If the evaluation returns a point to the left of the pivotal point (i.e., a feasible point), we rotate the direction of evaluation to the right, and evaluate the unconstrained curve again.
The rotation stops as soon as the point found makes the feasibility constraint binding, which means we have found the pivotal point.
To make this conceptual procedure practical, we replace the rotation with a binary search, which finds a point to the left of the pivotal point that is at most $\eps$ away (in terms of the polar angle) in $O(\log(1 / \eps))$ iterations.
When the constrained planning problem corresponds to the computation of SEFCE, we show that any two turning points on a Pareto frontier curve must be well separated, by some quantity that is at most exponentially small in the size of the game.
So, if we make $\eps$ exponentially small (which means there are polynomially many iterations in the binary search), the binary search is guaranteed to find the closest turning point to the left of the pivotal point.
We can then compute the pivotal point by finding the closest point to the right in the same way, and taking a convex combination of the two points found.

\paragraph{Bounding the number of evaluations.}
The above gives a recursive algorithm for evaluating Pareto frontier curves, but executed in the na\"ive way, the procedure may take exponential time in the size of the game.
We need one final observation to make the algorithm polynomial time in these parameters: For each $(s, a)$ pair, the pivotal point on the curve $\pf_{s, a}$ only needs to be evaluated once.
Given this, the total number of recursive evaluations triggered must be polynomial in the size of the game.
This is because new directions of evaluation emerge only when we binary search for a pivotal point.
Each binary search may create polynomially many new directions, and since we binary search for each pivotal point only once, there are $O(mn)$ binary searches in total, which means the total number of relevant directions is polynomial.
Moreover, each direction can only appear in $O(mn)$ evaluations (one for each state-action pair), since we never need to perform the same evaluation twice.
This means the total number of evaluations for all relevant directions is polynomial.

\subsection{Reduction to Constrained Planning}

We first provide a formal reduction from computing an SEFCE to the constrained planning problem.

\paragraph{The punishment amplifier.}
Fixing a stochastic game $(\cS, \cA, \ap, r_1, r_2, \trans)$, our reduction involves the punishment amplifier $\pa$, which maps each deterministic strategy to its maximally punishing version against a subset of players --- for SEFCE this subset is $\{2\}$, and as we will see later, for EFCE this subset is $\{1, 2\}$.
For each player $i \in \{1, 2\}$, consider the zero-sum stochastic game $(\cS, \cA, \ap, r_1', r_2', \trans)$, where $r_i' = r_i$ and $r_{3 - i}' = -r_i$.
Let $\pi_i$ be a deterministic subgame-perfect equilibrium strategy in this zero-sum game --- here, $i$ is the player being punished, but note that $\pi_i$ comprises both players' actions.
Such a strategy can be found by backward induction.
Note that without loss of generality, $\pi_i$ is history-independent, so we write $\pi_i(s)$ for simplicity.
If there are multiple candidates for $\pi_i$, we pick an arbitrary one among them (the choice does not affect our results).

Given a deterministic strategy $\pi$ and a subset of players $S \subseteq \{1, 2\}$, the punishment-amplified version $\pi' = \pa(\pi, S)$ of $\pi$ is given by: For each $h = (s_1, a_1, \dots, s_t, a_t) \in \cH$ and $s \in \cS$,
\begin{itemize}
    \item If $h \in \cH^\pi$ (i.e., $h$ is feasible under $\pi$), then $\pi'(h, s) = \pi(h, s)$.
    \item If $\ap(s) \notin S$, then $\pi'(h, s) = \pi(h, s)$.
    \item Otherwise, $\pi'(h, s) = \pi_i(s)$, where $h' = (s_1, a_1, \dots, s_{t'})$ is the longest feasible prefix of $h$, and $i = \ap(s_{t'})$ (i.e., $i$ is the first player who deviated, in state $s_{t'}$).
\end{itemize}
The punishment amplifier can be naturally extended to randomized strategies: For a randomized strategy $\Pi$ and a subset of players $S$, $\pa(\Pi, S)$ is obtained by mapping every deterministic strategy $\pi$ in the support of $\Pi$ to $\pa(\pi, S)$, and assign the latter the same probability mass in $\Pi'$ as $\pi$ has in $\Pi$.

We first prove maximum punishment is without loss of generality, which is formally captured by the following lemma:

\begin{lemma}
\label{lem:maximum_punishment}
    Fix a stochastic game $(\cS, \cA, \ap, r_1, r_2, \trans)$.
    For any $\Pi$ under which the follower is best responding, the follower is also best responding under $\pa(\Pi, \{2\})$, and moreover,
    \[
        u_1^{\Pi}(\emptyset, \si) = u_1^{\pa(\Pi, \{2\})}(\emptyset, \si).
    \]
\end{lemma}

We defer the proof of Lemma~\ref{lem:maximum_punishment}, as well as all other missing proofs, to Appendix~\ref{app:sefce}.
The lemma suggests that when optimizing over strategies where the follower is best responding, we can focus on those with the maximum punishment structure as described above.
Next we show this gives us a reduction to the constrained planning problem.

\begin{lemma}
\label{lem:utility_under_punishment}
    Fix a stochastic game $(\cS, \cA, \ap, r_1, r_2, \trans)$.
    For any randomized strategy $\Pi$, the follower is best responding under $\pa(\Pi, \{2\})$ if the following condition holds: For each admissible history $h \in \cH^\Pi$, state $s \in \cS$ where $\ap(s) = 2$, and action $a \in \cA$ such that $h + (s, a) \in \cH^\Pi$,
    \[
        v_2^{\Pi \mid (h + (s, a))}(h, s) \ge \max_{a' \in \cA} \left(r_2(s, a') + \bE_{s' \sim \trans(s, a')}[u_2^{\pi_2}(h + (s, a'), s')]\right),
    \]
    where $\pi_2$ is the subgame perfect equilibrium when the leader tries to minimize the follower's utility, as defined above.
\end{lemma}

Observe that in the above lemma, the right-hand side of the inequality does not depend on $\Pi$.
Moreover, since $\pi_i$ is history-independent, it does not depend on $h$ either.
So the right hand-side is a constant that depends only on the state $s$.
From now on, we call this quantity the utility under punishment in state $s$, defined as
\[
    u^p(s) = \max_{a' \in \cA} \left(r_i(s, a') + \bE_{s' \sim \trans(s, a')}[u_i^{\pi_i}(h + (s, a'), s')]\right),
\]
where $i = \ap(s)$ (note that although we define the utility under punishment for both players, for SEFCE we only need it for the follower, i.e., player $2$).
The utility under punishment $u^p(s)$ can be efficiently computed in all states.

Lemmas~\ref{lem:maximum_punishment}~and~\ref{lem:utility_under_punishment} together imply the following claim, which states that finding an SEFCE is equivalent to constrained planning.

\begin{theorem}
\label{thm:reduction}
    Fix a stochastic game $(\cS, \cA, \ap, r_1, r_2, \trans)$.
    For any $x \ge 0$, there exists a strategy $\Pi$ under which the follower is best responding such that $u_1^\Pi(\emptyset, \si) \ge x$, if and only if there exists a strategy $\Pi'$ such that $u_1^{\Pi'}(\emptyset, \si) \ge x$, and for each admissible history $h \in \cH^{\Pi'}$, state $s \in \cS$ where $\ap(s) = 2$, and action $a \in \cA$ such that $h + (s, a) \in \cH^{\Pi'}$,
    \[
        v_2^{\Pi' \mid (h + (s, a))}(h, s) \ge u^p(s).
    \]
\end{theorem}

\paragraph{Feasible strategies.}
We say a strategy $\Pi$ is feasible if it satisfies the condition in the theorem which involves the utility under punishment, i.e., for each $h \in \cH^\Pi$, $s \in \cS$ where $\ap(s) = 2$, and $a \in \cA$ such that $h + (s, a) \in \cH^\Pi$,
\[
    v_2^{\Pi \mid (h + (s, a))}(h, s) \ge u^p(s).
\]
We say a strategy $\Pi$ is feasible after state $s$, if for each $h \in \cH^\Pi$, $s' > s$ where $\ap(s') = 2$, and $a \in \cA$ such that $h + (s', a) \in \cH^\Pi$,
\[
    v_2^{\Pi \mid (h + (s', a))}(h, s') \ge u^p(s').
\]
Our problem now becomes finding a feasible strategy that maximizes player $1$'s utility.

\subsection{Pareto Frontier Curves}

Before proceeding to the full description of our algorithm, we first quickly (and somewhat informally) define Pareto frontier curves and discuss some useful properties.
Intuitively, the Pareto frontier curve $\pf_{s, a}$ for a state-action pair $(s, a)$ is the curve capturing all Pareto-optimal pairs of onward utilities (assuming an empty history) for both players after playing action $a$ in state $s$, induced by strategies that are feasible after $s$.
Another way to view $\pf_{s, a}$ is it is the top right boundary of the region of pairs of onward utilities (say $\fr_{s, a}$) induced by strategies that are feasible after $s$.
We call $\fr_{s, a}$ the feasible region for $(s, a)$, which can be defined in the following way:
\[
    \fr_{s, a} = \left\{(v_1^{\Pi \mid (s, a)}(\emptyset, s), v_2^{\Pi \mid (s, a)}(\emptyset, s)) \,\middle|\, \Pi\text{ is feasible after }s,\, (s, a) \in \cH^\Pi\right\}.
\]
We first argue that both $\fr_{s, a}$ and $\pf_{s, a}$ are well behaved:

\begin{lemma}
\label{lem:convex}
    Fix a stochastic game $(\cS, \cA, \ap, r_1, r_2, \trans)$.
    For any state $s \in \cS$ and action $a \in \cA$, $\fr_{s, a}$ is always a convex region and $\pf_{s, a}$ is always a concave curve.
\end{lemma}

Given Lemma~\ref{lem:convex}, $\pf_{s, a}$ in fact specifies a bijection between the two players' utilities under feasible strategies, corresponding to the two coordinates of a point.
In the rest of the paper, we will abuse notation, and use $\pf_{s, a}$ in $4$ ways:
\begin{itemize}
    \item For $x \in \bR_+$, $\pf_{s, a}(x) \in \bR_+$ denotes the $y$-coordinate of the point on $\pf_{s, a}$ whose $x$-coordinate is $x$.
    When used in this way, the argument to $\pf_{s, a}$ will always be $x$, possibly with subscripts or superscripts.
    \item For $y \in \bR_+$, $\pf_{s, a}(y) \in \bR_+$ denotes the $x$-coordinate of the point on $\pf_{s, a}$ whose $y$-coordinate is $y$.
    When used in this way, the argument to $\pf_{s, a}$ will always be $y$, possibly with subscripts or superscripts.
    \item For $p \in \bR_+^2$, we say $p \in \pf_{s, a}$ if $p$ is in $\pf_{s, a}$ as a set of points (i.e., the graph of $\pf_{s, a}$ as a mapping).
    \item For $\alpha \in \bR_+^2$, $\pf_{s, a}(\alpha) \in \bR_+^2$ denotes the farthest point on $\pf_{s, a}$ along direction $\alpha$.
    That is,
    \[
        \pf_{s, a}(\alpha) = \argmax_{p \in \pf_{s, a}} \alpha \cdot p.
    \]
\end{itemize}

\subsection{Evaluating the Pareto Frontier Curves}

Observe that if we can evaluate $\pf_{\si, a}$ for each $a \in \cA$, then it is not hard to find the optimal utilities of the two players induced by a feasible strategy, given a particular objective direction $\alpha \in \bR_+^2$ (for an SEFCE in particular, we want $\alpha = (1, 0)$):
Without loss of generality, an optimal strategy picks a deterministic action in the initial state $\si$, so we only need to try every one of the actions.
For each action $a \in \cA$, the optimal utilities induced by a strategy that is feasible after $\si$ is $\pf_{s, a}(\alpha)$.
If $\ap(\si) = 1$ or $\pf_{s, a}(\alpha)_{\ap(\si)} \ge u^p(\si)$, then this strategy is a feasible strategy, and $\pf_{s, a}(\alpha)$ gives the optimal utilities if the first action is $a$.
Otherwise, since $\pf_{s, a}$ is concave, the optimal utilities induced by a feasible strategy must correspond to the point on $\pf_{s, a}$ where the constraint in $\si$ is binding.
More specifically, suppose $\ap(\si) = 2$, and let $y_\si = u^p(\si)$.
Then the optimal utilities when the first action is $a$ must be $(\pf_{s, a}(y_\si), y_\si)$.

\paragraph{Pivotal points.}
The above discussion suggests that the point $(\pf_{s, a}(y_\si), y_\si)$ plays a particularly important role in finding the optimal utilities.
More generally, we define the pivotal point $\pp(s, a)$ on $\pf_{s, a}$ for each $(s, a)$ where $\ap(s) = 2$ to be the rightmost point (if there is one) on $\pf_{s, a}$ such that $\pp(s, a)_{\ap(s)} \ge u^p(s)$.
If such a point does not exist, then we let $\pp(s, a) = (-n, -n)$ (here, $-n$ is without loss of generality --- any quantity that is small enough would be consistent with our results).
For notational convenience, if $\ap(s) = 1$, we let $\pp(s, a) = \pf_{s, a}((1, 0))$.
Below we demonstrate how to evaluate the Pareto frontier curves recursively with the help of the pivotal points.

\begin{lemma}
\label{lem:evaluation}
    Fix a stochastic game $(\cS, \cA, \ap, r_1, r_2, \trans)$.
    For any $s \in \cS$, $a \in \cA$, and $\alpha \in \bR_+^2$,
    \[
        \pf_{s, a}(\alpha) = (r_1(s, a), r_2(s, a)) + \bE_{s' \sim \trans(s, a)}\left[p_{s'}\right],
    \]
    where for each $s' > s$ and $a' \in \cA$,
    \[
        p_{s', a'} = \begin{cases}
            \pf_{s', a'}(\alpha), & \mathrm{if}\ \ap(s') = 1\ \mathrm{or}\ \pf_{s', a'}(\alpha)_2 \ge u^p(s') \\
            \pp(s', a'), & \mathrm{otherwise},
        \end{cases}
    \]
    and for each $s' > s$,
    \[
        p_{s'} = \argmax_{p \in \{p_{s', a'}\}_{a' \in \cA}} p \cdot \alpha.
    \]
\end{lemma}

In words, the lemma says that once the pivotal points for all later state-action pairs have been computed, evaluating $\pf_{s, a}$ can be reduced to at most $mn$ evaluations of curves for later state-action pairs.
This reduction plays a central role in our algorithm.
Moreover, it also provides a way for bounding the numerical resolution of the Pareto frontier curves, which is captured by the following lemma.

\begin{lemma}
\label{lem:resolution}
    Fix a stochastic game $(\cS, \cA, \ap, r_1, r_2, \trans)$.
    Suppose all parameters of the game can be encoded using $L$ bits, i.e., for each $s, s' \in \cS$ and $a \in \cA$, $r_1(s, a)$, $r_2(s, a)$ and $\trans(s, a, s')$ are all multiples of $2^{-L}$.
    Then for each $s \in \cS \setminus \{\st\} = [n - 1]$, there exists some integer $C_s \le 2^{(n - s)(n + 1)L}$, such that for each $a \in \cA$:
    \begin{itemize}
        \item For any $\alpha \in \bR_+^2$, the $y$-coordinate of $\pf_{s, a}(\alpha)$ is a multiple of $2^{-(n - s)L}$, and the $x$-coordinate is a multiple of $1 / C_s$.
        \item If $\ap(s) = 2$, then the $y$-coordinate of $\pp(s, a)$ is a multiple of $2^{-(n - s)L}$, and the $x$-coordinate is a multiple of $1 / C_s$.
    \end{itemize}
    Moreover, for each $s \in [n - 2]$, $C_s$ is a multiple of $C_{s + 1}$.
\end{lemma}

We will use the following direct corollary of Lemma~\ref{lem:resolution}:

\begin{corollary}
\label{cor:resolution}
    Fix a stochastic game $(\cS, \cA, \ap, r_1, r_2, \trans)$.
    Suppose all parameters of the game can be encoded using $L$ bits.
    Then there exists an integer $C \le 2^{n^2 L}$ such that for each $s \in \cS$ and $a \in \cA$, both coordinates of $\pp(s, a)$ are multiples of $1 / C$, and for each $\alpha \in \bR_+^2$, both coordinates of $\pf_{s, a}(\alpha)$ are multiples of $1 / C$.
\end{corollary}

\subsection{Algorithm and Analysis}

Now we are ready to formally describe and analyze our full algorithm, Algorithm~\ref{alg:sefce}, which calls Algorithm~\ref{alg:eval} as a subroutine.

\begin{algorithm}[!ht]
\KwIn{a turn-taking stochastic game $(\cS = [n], \cA, \ap, r_1, r_2, \trans)$.}
\KwOut{the leader's utility under an SEFCE, together with an implicit representation of an SEFCE in the input game.}
    create a data structure $\cD$ that stores the results of all evaluations (used by $\eval$)\;
    \For{each state $s = n - 1, n - 2, \dots, 1$}{
        \If{$\ap(s) = 1$}{
            \For{each action $a \in \cA$}{
                let $p_{s, a} \gets \eval(s, a, (1, 0))$\;
            }
        }
        \Else{
            \For{each action $a \in \cA$}{
                let $\ell \gets (0, 1)$, $r \gets (1, 0)$, $q_\ell \gets \eval(s, a, \ell)$, $q_r \gets \eval(s, a, r)$\;
                if $(q_\ell)_2 < u^p(s)$, let $p_{s, a} \gets (-n, -n)$\;
                if $(q_r)_2 \ge u^p(s)$, let $p_{s, a} \gets q_r$\;
                \If{$(q_\ell)_2 \ge u^p(s)$ and $(q_r)_2 < u^p(s)$}{
                    \While{$\|\ell - r\|_1 \ge \frac{1}{3n \cdot 2^{2n^2 L}}$}{
                        let $q \gets \eval(s, a, (\ell + r) / 2)$ (see Algorithm~\ref{alg:eval})\;
                        let $\ell \gets (\ell + r) / 2$ if $q_2 \ge u^p(s)$, and $r \gets (\ell + r) / 2$ otherwise\;
                    }
                    let $q_\ell \gets \eval(s, a, \ell)$, $q_r \gets \eval(s, a, r)$, $\ell_{s, a} \gets \ell$, $r_{s, a} \gets r$\;
                    let $p_{s, a} \gets \left(\frac{(q_\ell)_2 - u^p(s)}{(q_\ell)_2 - (q_r)_2} \cdot (q_r)_1 + \frac{u^p(s) - (q_r)_2}{(q_\ell)_2 - (q_r)_2} \cdot (q_\ell)_1, u^p(s)\right)$\;
                }
            }
        }
    }
    let $\opt \gets \max_{a \in \cA}\ (p_{\si, a})_1$\;
    \Return $\opt$, $\{p_{s, a}\}_{s, a}$ $\{\ell_{s, a}\}_{s, a}$, $\{r_{s, a}\}_{s, a}$, and $\cD$\;
\caption{A polynomial-time algorithm for computing an SEFCE in turn-taking stochastic games.}
\label{alg:sefce}
\end{algorithm}

\begin{algorithm}[!ht]
\KwIn{a state $s$, an action $a$, a direction of evaluation $\alpha$, all variables in Algorthm~\ref{alg:sefce}.}
\KwOut{$\pf_{s, a}(\alpha)$.}
    if $s = \st = n$ then \Return $(0, 0)$\;
    \If{$(s, a, \alpha) \notin \cD$ (i.e., if $\cD(s, a, \alpha)$ does not exist)}{
        \For{$s' = s + 1, \dots, n$}{
            \For{$a' \in \cA$}{
                let $q_{s', a'} \gets \eval(s', a', \alpha)$\;
                \If{$\ap(s') = 2$ and $(q_{s', a'})_2 < u^p(s')$}{
                    let $q_{s', a'} \gets p_{s', a'}$\;
                }
            }
            let $q_{s'} \gets \argmax_{q \in \{q_{s', a'}\}_{a' \in \cA}} \alpha \cdot q$\;
        }
        let $\cD(s, a, \alpha) \gets (r_1(s, a), r_2(s, a)) + \bE_{s' \sim \trans(s, a)}[q_{s'}]$\;
    }
    \Return $\cD(s, a, \alpha)$\;
\caption{$\eval$: A subroutine of Algorithm~\ref{alg:sefce} that performs recursive evaluations as needed.}
\label{alg:eval}
\end{algorithm}

Below we analyze our algorithm.
First we show that the binary search in Algorithm~\ref{alg:sefce} is exact, in the sense that in line~18, $q_\ell$ and $q_r$ are adjacent turning points to each other on $\pf_{s, a}$.

\begin{lemma}
\label{lem:binary_search}
    Fix a stochastic game $(\cS, \cA, \ap, r_1, r_2, \trans)$, where all parameters of the game can be encoded using $L$ bits.
    In every execution of line~18 of Algorithm~\ref{alg:sefce}, assuming $q_\ell = \pf_{s, a}(\ell)$ and $q_r = \pf_{s, a}(r)$ (we will prove this later), $q_\ell$ and $q_r$ are adjacent turning points to each other on $\pf_{s, a}$.
\end{lemma}

Now we are ready to prove the key property of the algorithm, which is captured by the following lemma.

\begin{lemma}
\label{lem:key}
    Fix a stochastic game $(\cS, \cA, \ap, r_1, r_2, \trans)$, where all parameters of the game can be encoded using $L$ bits.
    The following statements regarding Algorithms~\ref{alg:sefce}~and~\ref{alg:eval} hold:
    \begin{itemize}
        \item For each $s \in \cS$, $a \in \cA$ and $\alpha \in \bR_+^2$, if $(s, a, \alpha) \in \cD$ then $\cD(s, a, \alpha) = \pf_{s, a}(\alpha)$.
        \item For each $s \in \cS$ and $a \in \cA$, $p_{s, a}$ computed in Algorithm~\ref{alg:sefce} is the same as $\pp(s, a)$.
    \end{itemize}
\end{lemma}
\begin{proof}
    Apply induction on $s$.
    When $s = n - 1$, it is easy to check $p_{s, a} = \pp(s, a)$ and $\cD(s, a, \alpha) = \pf_{s, a}(\alpha)$.
    Now suppose the statements hold for all $s' > s$.
    Consider the first bullet point.
    For each $a \in \cA$ and $\alpha \in \bR_+^2$, observe that if $(s, a, \alpha) \in \cD$, then it is computed precisely in the way given in Lemma~\ref{lem:evaluation}.
    Given the induction hypothesis, this implies $\cD(s, a, \alpha) = \pf_{s, a}(\alpha)$.

    As for the second bullet point, consider $4$ cases:
    \begin{itemize}
        \item $\ap(s) = 1$.
        The first bullet point immediately implies $p_{s, a} = \pp(s, a)$. 
        \item $\ap(s) = 2$ and $(q_\ell)_2 < u^p(s)$ in line 11.
        Given the first bullet point, this means there is no point on $\pf_{s, a}$ whose $y$-coordinate is at least $u^p(s)$, and by definition, $\pp(s, a) = (-n, -n) = p_{s, a}$.
        \item $\ap(s) = 2$ and $(q_r)_2 \ge u^p(s)$ in line 12.
        Given the first bullet point, this means the entire $\pf_{s, a}$ is above $y = u^p(s)$, and by definition, $\pp(s, a) = \pf_{s, a}(0, 1) = p_{s, a}$.
        \item $\ap(s) = 2$, $(q_\ell)_2 \ge u^p(s)$, and $(q_r)_2 < u^p(s)$.
        This is the case where the binary search is executed.
        By Lemma~\ref{lem:binary_search} (and also given the first bullet point), $q_\ell$ and $q_r$ are adjacent turning points on $\pf_{s, a}$.
        Moreover, $(q_\ell)_2 \ge u^p(s)$ and $(q_r)_2 < u^p(s)$.
        Then $\pp(s)$ must be the unique convex combination of $q_\ell$ and $q_r$ whose $y$-coordinate is precisely $u^p(s)$, which is $p_{s, a}$ computed in line 19.
    \end{itemize}
\end{proof}

Now we can put everything together and prove the correctness and efficiency of Algorithm~\ref{alg:sefce} (we will discuss how to decode the output of Algorithm~\ref{alg:sefce} momentarily).

\begin{theorem}
\label{thm:sefce}
    Algorithm~\ref{alg:sefce} computes the leader's (player $1$'s) utility in an SEFCE in time polynomial in $n$, $m$, and $L$.
\end{theorem}
\begin{proof}
    For correctness: By Lemma~\ref{lem:key}, for each $a \in \cA$, $p_{\si, a} = \pp(\si, a)$, so $\opt = \max_a (p_{\si, a})_1 = \max_a \pp(\si, a)_1$ is player $1$'s optimal utility induced by a feasible strategy (where player $2$ is best responding), which is player $1$'s utility in an SEFCE.

    For efficiency: We only need to bound the number of times that $\eval$ is called.
    Observe that the number of times that $\eval$ is called in Algorithm~\ref{alg:sefce} is $O(nm \log(n 2^{n^2 L})) = O(n^2 m L \log n) = \mathrm{poly}(n, m, L)$.
    As for recursive calls, observe that $\eval$ makes $O(nm)$ recursive calls only when $(s, a, \alpha)$ is not in $\cD$ yet.
    So each tuple $(s, a, \alpha)$ may trigger $O(nm)$ calls in $\eval(s, a, \alpha)$.
    Let $A = \{\alpha \mid (s, a, \alpha) \in \cD\}$.
    Then the total number of recursive calls is at most $|\{(s, a, \alpha) \mid s \in \cS, a \in \cA, \alpha \in A\}| = O(nm|A|)$, so we only need to bound $|A|$.
    To this end, observe that for each $\alpha \in A$, there must be some $s \in \cS$ and $a \in \cA$ such that the binary search for $(s, a)$ involves $\alpha$.
    Each binary search involves $O(\log(n 2^{n^2 L})) = \mathrm{poly}(n, L)$ directions, so the total number of directions involved in these binary searches is $\mathrm{poly}(n, m, L)$.
    The latter is an upper bound of $|A|$.
\end{proof}

\subsection{Decoding the Output Strategy}

Now we discuss the final missing piece of our algorithm: extracting the strategy encoded in the output of Algorithm~\ref{alg:sefce}.
We present a procedure, Algorithm~\ref{alg:decode}, which, given the output of Algorithm~\ref{alg:sefce}, computes a random action for any given history-state pair.
We will prove that the strategy implicitly given by Algorithm~\ref{alg:decode} is the one encoded in the output of Algorithm~\ref{alg:sefce}.
In particular, it is feasible, and achieves the leader's utility in an SEFCE computed by Algorithm~\ref{alg:sefce}.

\begin{algorithm}[!ht]
\KwIn{A turn-taking stochastic game $(\cS, \cA, \ap, r_1, r_2, \trans)$, the output of Algorithm~\ref{alg:sefce}, a history $h = (s_1, a_1, \dots, s_t, a_t)$, and a state $s$.}
\KwOut{$\pi(h, s)$, where $\pi \sim \Pi \mid h$, and $\Pi$ is the strategy encoded in the output of Algorithm~\ref{alg:sefce}; $\error$ if $h$ is not an admissible history under $\Pi$.}
    let $a \gets \argmax_{a' \in \cA} (p_{s, a'})_1$, $\alpha \gets (1, 0)$, $q \gets p_{s, a}$\;
    if $|h| = 0$ \Return $a$\;
    if $a_1 \ne a$, \Return $\error$\;
    \For{$i = 1, 2, \dots, t - 1$}{
        \If{$\ap(s_i) = 2$ and $q = p_{s_i, a_i}$}{
            let $\ell \gets \ell_{s_i, a_i}$, $r \gets r_{s_i, a_i}$, $a_\ell \gets \argmax_{a' \in \cA} \ell \cdot \max^2\{p_{s_{i + 1}, a'}, \cD(s_{i + 1}, a', \ell)\}$, $a_r \gets \argmax_{a' \in \cA} r \cdot \max^2\{p_{s_{i + 1}, a'}, \cD(s_{i + 1}, a', \ell)\}$\;
            \tcc{for two points $q_1$ and $q_2$, $\max^k\{q_1, q_2\}$ denotes the point with the larger $k$-th coordinate between the two}
            let $\alpha \gets \ell$ if $a_{i + 1} = a_\ell$\;
            let $\alpha \gets r$ if $a_{i + 1} = a_r$\;
            if $a_{i + 1} \notin \{a_\ell, a_r\}$, \Return $\error$\;
        }
        \Else{
            let $a \gets \argmax_{a' \in \cA} \alpha \cdot \max^2\{p_{s_{i + 1}, a'}, \cD(s_{i + 1}, a', \alpha)\}$\;
            if $a_{i + 1} \ne a$ \Return $\error$\;
        }
        let $q \gets \max^2\{p_{s_{i + 1}, a_{i + 1}}, \cD(s_{i + 1}, a_{i + 1}, \alpha)\}$\;
    }
    \If{$\ap(s_t) = 2$ and $q = p_{s_t, a_t}$}{
        let $\ell \gets \ell_{s_i, a_i}$, $a_\ell \gets \argmax_{a' \in \cA} \ell \cdot \max^2\{p_{s, a'}, \cD(s, a', \ell)\}$, $q_\ell \gets \max^2\{p_{s, a_\ell}, \cD(s, a_\ell, \ell)\}$;
        let $r \gets r_{s_i, a_i}$, $a_r \gets \argmax_{a' \in \cA} r \cdot \max^2\{p_{s, a'}, \cD(s, a', r)\}$, $q_r \gets \max^2\{p_{s, a_r}, \cD(s, a_r, r)\}$\;
        let $a \gets a_\ell$ with probability $\frac{q_r - q}{q_r - q_\ell}$, $a \gets a_r$ with probability $\frac{q - q_\ell}{q_r - q_\ell}$\;
    }
    \Else{
        let $a \gets \argmax_{a' \in \cA} \alpha \cdot \max^2\{p_{s, a'}, \cD(s, a', \alpha)\}$\;
    }
    \Return $a$\;
\caption{A procedure that decodes the output of Algorithm~\ref{alg:sefce}.}
\label{alg:decode}
\end{algorithm}

\begin{theorem}
\label{thm:decode}
    Algorithm~\ref{alg:decode} outputs a feasible strategy $\Pi$ (restricted to admissible histories), which satisfies $u_1^\Pi(\emptyset, \si) = \opt$, where $\opt$ is the leader's utility in an SEFCE computed by Algorithm~\ref{alg:sefce}.
\end{theorem}
\begin{proof}
    First observe that Algorithm~\ref{alg:decode} does output a strategy restricted to admissible histories.
    In fact, it specifies a random action for each history-state pair, which can be viewed as a Bayesian description of a randomized strategy $\Pi$.
    We need to show that $u_1^\Pi(\emptyset, \si) = \opt$, and for each $h \in \cH^\Pi$, $s \in \cS$ where $\ap(s) = 2$, and $a$ such that $h + (s, a) \in \cH^\Pi$,
    \[
        v_2^{\Pi \mid h + (s, a)}(h, s) \ge u^p(s).
    \]
    To this end, observe that the output strategy faithfully implements the corresponding point on the corresponding Pareto frontier curve.
    That is, for each $i \in \{1, 2\}$, $h \in \cH^\Pi$, $s \in \cS$ where $\ap(s) = 2$, and $a$ such that $h + (s, a) \in \cH^\Pi$,
    \[
        v_i^{\Pi \mid h + (s, a)}(h, s) = \max^2\{p_{s, a}, \cD(s, a, \alpha)\} = \max^2\{\pp(s, a), \pf_{s, a}(\alpha)\}.
    \]
    (Recall that for two points $q_1$ and $q_2$, $\max^2\{q_1, q_2\}$ denotes the point with the larger $y$-coordinate between the two.)
    This can be proved inductively, and we omit the details (which are already quite repetitive at this point).
    Given the above correspondence, for the first condition, $u_1^\Pi(\emptyset, \si) = \max_{a' \in \cA} (p_{s, a'})_1$, which is equal to $\opt$.
    For the second condition, observe that $\max^2\{\pp(s, a), \pf_{s, a}(\alpha)\}$ is always a feasible point whose $y$-coordinate is at least $u^p(s)$, whenever $\ap(s) = 2$.
    This completes the proof.
\end{proof}

Finally, note that Algorithm~\ref{alg:decode} only specifies actions for admissible histories.
For inadmissible histories, both players should follow the equilibrium $\pi_2$ that maximally punishes player $2$ defined earlier.

\section{Approximately Optimal Extensive-Form Correlated Equilibria}
\label{sec:efce}

Now we proceed to the computation of approximately optimal EFCE.
We present a bi-criteria algorithm that, given an objective direction (i.e., a combination of the two players' utilities), computes an $\eps$-EFCE whose objective value is at least that of the optimal EFCE minus $\eps$, in time $\log(1 / \eps)$.
The idea and structure of our algorithm for approximately optimal EFCE is overall quite similar to that for SEFCE.
There are two key differences:
\begin{itemize}
    \item Recall that for SEFCE, we optimize over strategies where player $2$ is best responding.
    This reduces to optimizing over feasible strategies, where feasibility means that when player $2$ is the acting player, their onward utility must be at least the utility under punishment.
    For $\eps$-EFCE, both players need to be $\eps$-best responding, which leads to a different definition for feasible strategies.
    The definition and structural properties of Pareto frontier curves also need to be modified accordingly.
    Such modifications lead to minor changes in the proofs of the structural properties and the algorithm.
    \item A more substantial difference is in the numerical resolution of the Pareto frontier curves.
    For SEFCE, the feasibility constraints are all in the same direction, i.e., parallel to the $x$-axis.
    This is no longer true for $\eps$-EFCE, where the direction of the feasibility constraint in a state depends on the acting player.
    Such alternating constraints break the asymmetry between the two axes, which was crucial in the analysis of the numerical resolution of the Pareto frontier curves.
    As a result, a binary search with polynomially many iterations is no longer guaranteed to find the pivotal point exactly.
    Instead, the guarantee we have is that the error diminishes exponentially fast as the number of iterations grows.
    Importantly, this means inaccuracy in terms of both the objective value and the feasibility constraints.
    A careful analysis shows that the inaccuracy does not blow up too much as we approximately evaluate the Pareto frontier curves recursively.
\end{itemize}

\subsection{Useful Facts}

Before stating the full algorithm, we quickly state the new reduction from $\eps$-EFCE to constrained planning, as well as modified definitions and properties of Pareto frontier curves.
The proofs of these properties are similar to those of their counterparts for SEFCE.

\paragraph{Reduction to constrained planning.}

\begin{lemma}
\label{lem:maximum_punishment_efce}
    Fix a stochastic game $(\cS, \cA, \ap, r_1, r_2, \trans)$.
    For any $\eps \ge 0$ and $\Pi$ under which both players are $\eps$-best responding, both players are also $\eps$-best responding under $\pa(\Pi, \{1, 2\})$, and moreover, for each $i \in \{1, 2\}$,
    \[
        u_i^{\Pi}(\emptyset, \si) = u_i^{\pa(\Pi, \{1, 2\})}(\emptyset, \si).
    \]
\end{lemma}

\begin{lemma}
\label{lem:utility_under_punishment_efce}
    Fix a stochastic game $(\cS, \cA, \ap, r_1, r_2, \trans)$.
    For any $\eps \ge 0$ and randomized strategy $\Pi$, both players are $\eps$-best responding under $\pa(\Pi, \{1, 2\})$ if the following condition holds: For each admissible history $h \in \cH^\Pi$, state $s \in \cS$, and action $a \in \cA$ such that $h + (s, a) \in \cH^\Pi$,
    \[
        v_{\ap(s)}^{\Pi \mid (h + (s, a))}(h, s) \ge \max_{a' \in \cA} \left(r_{\ap(s)}(s, a') + \bE_{s' \sim \trans(s, a')}[u_{\ap(s)}^{\pi_{\ap(s)}}(h + (s, a'), s')]\right) - \eps,
    \]
    where $\pi_{\ap(s)}$ is the subgame perfect equilibrium when player $3 - \ap(s)$ tries to minimize player $\ap(s)$'s utility, as defined above.
\end{lemma}

\begin{theorem}
\label{thm:reduction_efce}
    Fix a stochastic game $(\cS, \cA, \ap, r_1, r_2, \trans)$.
    For any $\eps \ge 0$ and $(x, y) \in \bR_+^2$, there exists a strategy $\Pi$ under which both players are $\eps$-best responding such that $u_1^\Pi(\emptyset, \si) \ge x$ and $u_2^\Pi(\emptyset, \si) \ge y$, if and only if there exists a strategy $\Pi'$ such that $u_1^{\Pi'}(\emptyset, \si) \ge x$, $u_2^{\Pi'}(\emptyset, \si) \ge y$, and for each admissible history $h \in \cH^{\Pi'}$, state $s \in \cS$, and action $a \in \cA$ such that $h + (s, a) \in \cH^{\Pi'}$,
    \[
        v_{\ap(s)}^{\Pi' \mid (h + (s, a))}(h, s) \ge u^p(s) - \eps.
    \]
\end{theorem}

\paragraph{Feasible strategies.}
We say a strategy $\Pi$ is $\eps$-feasible (we will omit $\eps$ when it is clear from the context) if it satisfies the condition in the corollary which involves the utility under punishment, i.e., for each $h \in \cH^\Pi$, $s \in \cS$, and $a \in \cA$ such that $h + (s, a) \in \cH^\Pi$,
\[
    v_{\ap(s)}^{\Pi \mid (h + (s, a))}(h, s) \ge u^p(s) - \eps.
\]
We say a strategy $\Pi$ is $\eps$-feasible after state $s$, if for each $h \in \cH^\Pi$, $s' > s$, and $a \in \cA$ such that $h + (s', a) \in \cH^\Pi$,
\[
    v_{\ap(s')}^{\Pi \mid (h + (s', a))}(h, s') \ge u^p(s') - \eps.
\]

\paragraph{Pareto frontier curves.}
Again, we define the Pareto frontier curve $\pf_{s, a}$ (dependence on $\eps$ omitted) for a state-action pair $(s, a)$ to be the curve capturing all Pareto-optimal pairs of onward utilities (assuming an empty history) for both players after playing action $a$ in state $s$, induced by strategies that are feasible after $s$.

\begin{lemma}
\label{lem:convex_efce}
    Fix a stochastic game $(\cS, \cA, \ap, r_1, r_2, \trans)$.
    For any state $s \in \cS$ and action $a \in \cA$, $\pf_{s, a}$ is always a concave curve.
\end{lemma}

\paragraph{Pivotal points.}
Given an objective direction $\alpha_\obj$, we define the pivotal point $\pp(s, a)$ on $\pf_{s, a}$ for each $(s, a)$ to be the farthest point (if there is one) along $\alpha_\obj$ on $\pf_{s, a}$ such that $\pp(s, a)_{\ap(s)} \ge u^p(s)$.
If such a point does not exist, then we let $\pp(s, a) = (-n, -n)$ (again, $-n$ is without loss of generality, and any quantity that is small enough would be consistent with our results).

\begin{lemma}
\label{lem:evaluation_efce}
    Fix a stochastic game $(\cS, \cA, \ap, r_1, r_2, \trans)$.
    For any $s \in \cS$, $a \in \cA$, and $\alpha \in \bR_+^2$,
    \[
        \pf_{s, a}(\alpha) = (r_1(s, a), r_2(s, a)) + \bE_{s' \sim \trans(s, a)}\left[p_{s'}\right],
    \]
    where for each $s' > s$ and $a' \in \cA$,
    \[
        p_{s', a'} = \begin{cases}
            \pf_{s', a'}(\alpha), & \mathrm{if}\ \pf_{s', a'}(\alpha)_{\ap(s')} \ge u^p(s') \\
            \pp(s', a'), & \mathrm{otherwise},
        \end{cases}
    \]
    and for each $s' > s$,
    \[
        p_{s'} = \argmax_{p \in \{p_{s', a'}\}_{a' \in \cA}} p \cdot \alpha.
    \]
\end{lemma}

For EFCE, we need an approximate version of Lemma~\ref{lem:evaluation_efce}, which roughly says if we can approximately compute the pivotal points for all later state-action pairs, then approximately evaluating $\pf_{s, a}$ can be reduced to at most $mn$ evaluations of curves for later state-action pairs.
This is captured by the following claim, which is a direct corollary of Lemma~\ref{lem:evaluation_efce}.

\begin{corollary}
\label{cor:approximate_evaluation}
    Fix a stochastic game $(\cS, \cA, \ap, r_1, r_2, \trans)$.
    For any $s \in \cS$, $a \in \cA$, and $\alpha \in \bR_+^2$,
    \[
        \alpha \cdot \pf_{s, a}(\alpha) \le \alpha \cdot \left((r_1(s, a), r_2(s, a)) + \bE_{s' \sim \trans(s, a)}\left[q_{s'}\right]\right) + \eps,
    \]
    where for each $s' > s$, $q_{s'}$ satisfies
    \[
        \alpha \cdot p_{s'} \le \alpha \cdot q_{s'} + \eps.
    \]
    Here, $p_{s'}$ is defined in the same way as in Lemma~\ref{lem:evaluation_efce}.
\end{corollary}

\subsection{Algorithm and Analysis}

Now we are ready to present and analyze our algorithm for approximately optimal EFCE, Algorithm~\ref{alg:efce}, which uses Algorithm~\ref{alg:approximate_eval} as a subroutine.
We defer both these algorithms, as well as Algorithm~\ref{alg:decode_efce} to be mentioned later, to Appendix~\ref{app:efce}, since these algorithms are similar to their counterparts in Section~\ref{sec:sefce}.

The key differences between Algorithms~\ref{alg:efce}~and~\ref{alg:approximate_eval}, and Algorithms~\ref{alg:sefce}~and~\ref{alg:eval}, are:
\begin{itemize}
    \item Now we need to satisfy feasibility constraints in all states, whereas for SEFCE constraints exist only in states where the acting player is the follower.
    \item The binary search stops when the two directions are $\eps / n$-close to each other, and in general, it only finds an approximate pivotal point as opposed to an exact one.
    Accordingly, we also allow inaccuracy in the feasibility constraints.
\end{itemize}

In order to analyze the algorithm, we first show that the binary search finds a point that is close to the actual pivotal point.
We defer the proof of Lemma~\ref{lem:approximate_binary_search}, as well as that of Lemma~\ref{lem:key_efce} below, to Appendix~\ref{app:efce}.
For an illustration of the proof of Lemma~\ref{lem:approximate_binary_search}, see Figure~\ref{fig:approximate_binary_search}.

\begin{lemma}
\label{lem:approximate_binary_search}
    Fix a stochastic game $(\cS, \cA, \ap, r_1, r_2, \trans)$. 
    In Algorithm~\ref{alg:efce}, assuming for each $s \in \cS$, $a \in \cA$ and $\alpha \in \bR_+^2$ where $\|\alpha\|_1 \le 1$, $\cD(s, a, \alpha)$ satisfies $\alpha \cdot \cD(s, a, \alpha) \ge \alpha \cdot \pf_{s, a}(\alpha) - \frac{n - s - 1}{n} \cdot \eps$, $p_{s, a}$ computed in line~15 satisfies $(p_{s, a})_{3 - \ap(s)} \ge \pp(s, a)_{3 - \ap(s)} - \frac{n - s}{n} \cdot \eps$.
\end{lemma}

\begin{figure}[t]
    \centering
    \includegraphics[width=0.48\linewidth]{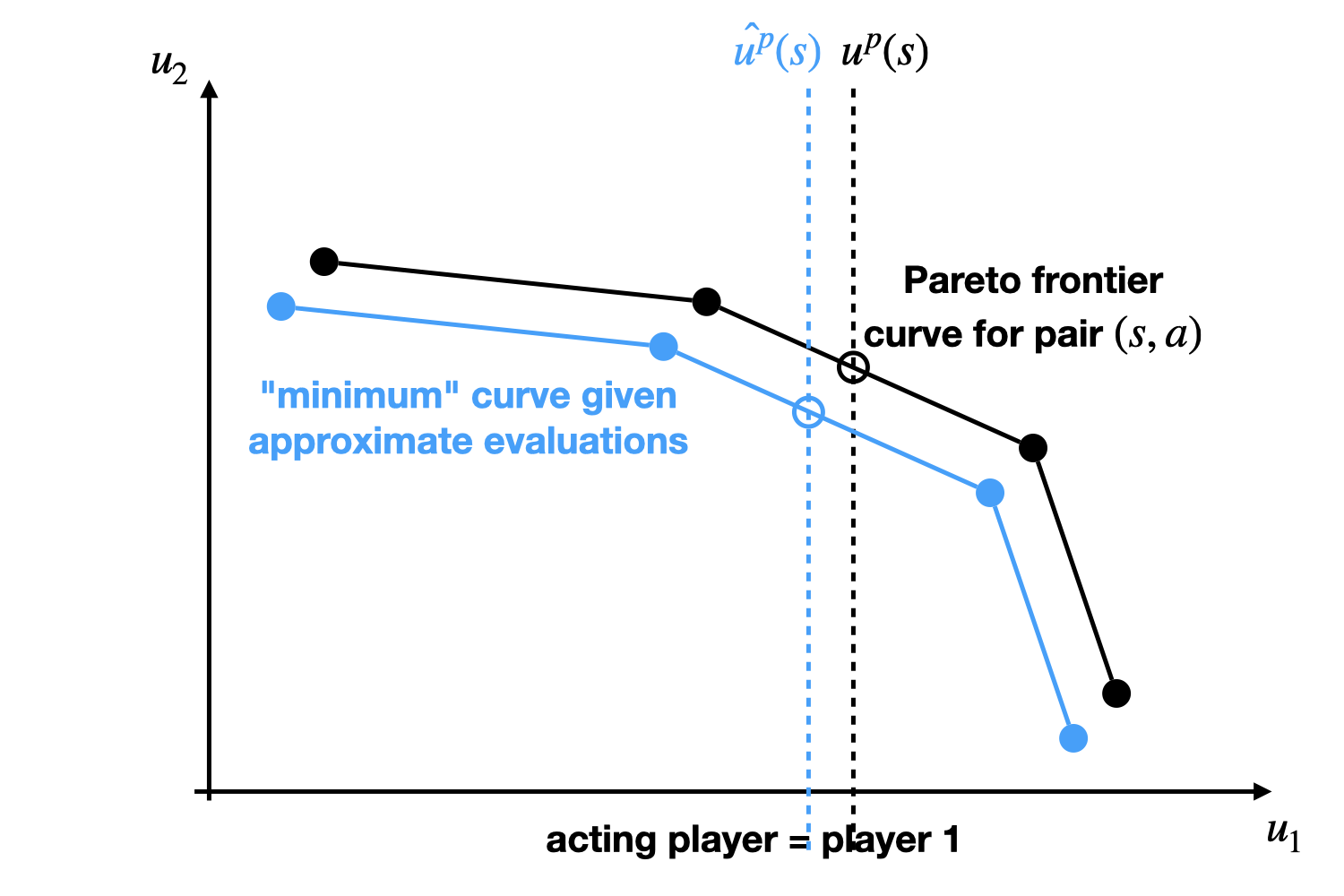}
    \includegraphics[width=0.48\linewidth]{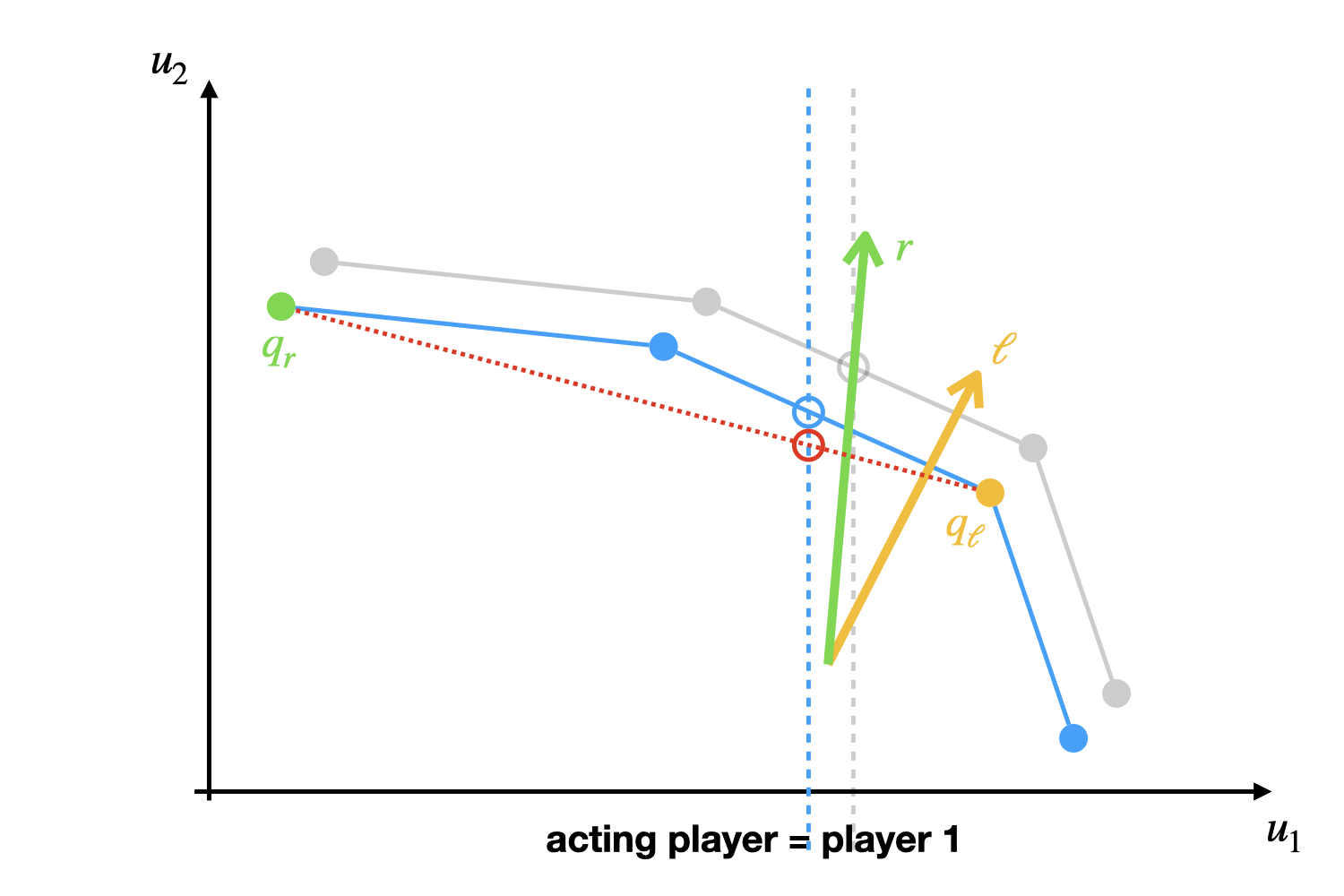}
    \caption{Illustration of the proof of Lemma~\ref{lem:approximate_binary_search}.}
    \label{fig:approximate_binary_search}
\end{figure}

We then establish the key properties of Algorithm~\ref{alg:efce} --- essentially an approximate version of Lemma~\ref{lem:key}.

\begin{lemma}
\label{lem:key_efce}
    Fix a stochastic game $(\cS, \cA, \ap, r_1, r_2, \trans)$.
    The following statements regarding Algorithms~\ref{alg:efce}~and~\ref{alg:approximate_eval} hold: For each $s \in \cS$, $a \in \cA$ and $\alpha \in \bR_+^2$ where $\|\alpha\|_1 = 1$,
    \begin{itemize}
        \item if $(s, a, \alpha) \in \cD$, then $\alpha \cdot \cD(s, a, \alpha) \ge \alpha \cdot \pf_{s, a}(\alpha) - \frac{n - s - 1}{n} \cdot \eps$;
        \item $\alpha \cdot (p_{s, a}) \ge \alpha \cdot \pp(s, a) - \frac{n - s}{n} \cdot \eps$.
    \end{itemize}
\end{lemma}

Given Lemma~\ref{lem:key_efce}, it is not hard to prove the correctness and efficiency of Algorithm~\ref{alg:efce}.

\begin{theorem}
\label{thm:efce}
    Algorithm~\ref{alg:efce} runs in time polynomial in $n$, $m$, and $\log(1 / \eps)$ (assuming all arithmetic operations take constant time), and the output $\opt$ satisfies:
    \begin{itemize}
        \item $\opt$ is smaller than the optimal objective value of any EFCE by at most $\eps$, i.e.,
        \[
            \opt \ge \max_{\Pi\ \mathrm{is\ an\ EFCE}} \alpha_\obj \cdot (u_1^\Pi(\emptyset, \si), u_2^\Pi(\emptyset, \si)) - \eps.
        \]
        \item There exists an $\eps$-EFCE whose objective value is $\opt$.
    \end{itemize}
\end{theorem}

The proof of the above theorem is similar to that of Theorem~\ref{thm:sefce}, with one exception: We prove the second bullet point by giving an algorithm that constructs an $\eps$-EFCE whose objective value is $\opt$.
The algorithm (Algorithm~\ref{alg:decode_efce} in Appendix~\ref{app:efce}) is overall quite similar to Algorithm~\ref{alg:decode}.
As such, Algorithm~\ref{alg:decode_efce} only specifies actions for admissible histories.
For inadmissible histories, both players should follow the equilibrium $\pi_i$ that maximally punishes player $i$ defined earlier, where $i$ is the player who first deviates.
The proof of the following claim is similar to that of Theorem~\ref{thm:decode}.

\begin{theorem}
\label{thm:decode_efce}
    Algorithm~\ref{alg:decode_efce} outputs a feasible strategy $\Pi$ (restricted to admissible histories), which satisfies $\alpha_\obj \cdot (u_1^\Pi(\emptyset, \si), u_2^\Pi(\emptyset, \si)) = \opt$, where $\opt$ is the approximately optimal objective value computed by Algorithm~\ref{alg:efce}.
\end{theorem}

\section*{Acknowledgments}
    Zhang and Conitzer thank the Cooperative AI Foundation, Polaris Ventures (formerly the Center for Emerging Risk Research) and Jaan Tallinn's donor-advised fund at Founders Pledge for financial support.
    Cheng is supported in part by NSF Award CCF-2307106.

\bibliographystyle{plainnat}
\bibliography{ref}

\appendix
\input{app}

\end{document}

%% file: app.tex
\section{Omitted Proofs in Section~\ref{sec:sefce}}
\label{app:sefce}

\begin{proof}[Proof of Lemma~\ref{lem:maximum_punishment}]
    For brevity let $\Pi' = \pa(\Pi, \{2\})$.
    We first prove $\Pi$ and $\Pi'$ induce the same utility for the leader.
    Observe that $\Pi$ and $\Pi'$ share the same set of admissible histories, i.e., $\cH^\Pi = \cH^{\Pi'}$.
    Moreover, for any $h \in \cH^\Pi$, $\pa(\Pi \mid h, \{2\}) = \Pi' \mid h$.
    Given the above, we have a stronger claim: For each $i \in \{1, 2\}$, $h$ and $h' \in \cH^\Pi$, and $s \in \cS$,
    \[
        v_i^{\Pi \mid h}(h', s) = v_i^{\Pi' \mid h}(h', s).
    \]
    This can be verified by expanding both sides using the definition of $v_i$, coupling $\pi \sim \Pi \mid h$ with $\pi' = \pa(\pi, \{2\}) \sim \Pi' \mid h$, and checking $\pi$ and $\pi'$ always induce the same play given $h'$ and $s$ since $h'$ is admissible.
    Setting $h = h' = \emptyset$ and $s = \si$, this immediately implies that both players have the same utilities under $\Pi$ and $\Pi'$.

    We now prove that the follower is best responding under $\Pi'$.
    Consider any admissible history $h \in \cH^\Pi$, state $s$ where $\ap(s) = 2$, action $a$ where $h + (s, a) \in \cH^\Pi$, and deterministic strategy $\pi''$ where $\pi''(h, s) \ne a$ (we reserve $\pi'$ for later use).
    Since player $2$ is best responding under $\Pi$, by definition we have
    \[
        v_2^{\Pi \mid (h + (s, a))}(h, s) \ge \bE_{\pi \sim \Pi \mid (h + (s, a))}\left[u_2^{(2: \pi'', 1: \pi)}(h, s)\right]. 
    \]
    We already know that
    \[
        v_2^{\Pi \mid (h + (s, a))}(h, s) = v_2^{\Pi' \mid (h + (s, a))}(h, s).
    \]
    So we only need to prove that
    \[
        \bE_{\pi \sim \Pi \mid (h + (s, a))}\left[u_2^{(2: \pi'', 1: \pi)}(h, s)\right] \ge \bE_{\pi' \sim \Pi' \mid (h + (s, a))}\left[u_2^{(2: \pi'', 1: \pi')}(h, s)\right].
    \]
    Again we couple $\pi \sim \Pi \mid (h + (s, a))$ with $\pi' = \pa(\pi) \sim \Pi' \mid (h + (s, a))$, so we only need to compare $u_2^{(2: \pi'', 1: \pi)}(h, s)$ and $u_2^{(2: \pi'', : \pi')}(h, s)$.
    Both quantities involve summing over the rewards in up to $n$ steps and taking expectations over random transitions.
    To this end, observe that the first steps in both quantities are always the same (player $2$ playing $\pi''(h, s)$), so we further couple them.
    Now we only need to prove in the subgame induced by $h + (s, \pi''(h, s))$ and $s' \sim \trans(s, \pi''(h, s))$, player $2$'s utility when the other player follows $\pi$ is at least player $2$'s utility when the other player follows $\pi'$.
    This follows almost directly from the definition of $\pi'$: Restricted to this subgame and player $1$, $\pi'$ behaves identically as $\pi_2$, which is a subgame-perfect equilibrium when the other player tries to minimize player $2$'s utility.
    In other words, fixing player $2$'s strategy (which is $\pi''$), player $2$'s utility in this subgame against $\pi$ is no smaller than player $2$'s utility against $\pi'$.
    Now taking the expectations over $\pi$, $\pi'$, and $s'$ gives the desired inequality.
\end{proof}

\begin{proof}[Proof of Lemma~\ref{lem:utility_under_punishment}]
    Let $\Pi' = \ap(\Pi, \{2\})$.
    We only need to verify that if $\Pi$ satisfies the condition in the lemma, then for any admissible history $h \in \cH^{\Pi'}$, state $s$ where $\ap(s) = 2$, action $a$ where $h + (s, a) \in \cH^{\Pi'}$, and deterministic strategy $\pi''$ where $\pi''(h, s) \ne a$,
    \[
        v_2^{\Pi \mid (h + (s, a))}(h, s) \ge \bE_{\pi' \sim \Pi' \mid (h + (s, a))}\left[u_2^{(2: \pi'', 1: \pi')}(h, s)\right].
    \]
    In particular, we only need to show that the right hand side of the above inequality is upper bounded by
    \[
        \max_{a' \in \cA} \left(r_2(s, a') + \bE_{s' \sim \trans(s, a')}[u_2^{\pi_2}(h + (s, a'), s')]\right).
    \]
    Again, this follows almost directly from the definition of $\pi'$: Restricted to this subgame and player $1$, $\pi'$ behaves identically as $\pi_2$, which is a subgame-perfect equilibrium when the other player tries to minimize player $2$'s utility.
    In other words, player $2$'s utility in this subgame against $\pi'$ is at most
    \[
        r_2(s, a'') + \bE_{s'' \sim \trans(s, a'')}[u_2^{\pi_2}(h + (s, a''), s'')],
    \]
    where $a'' = \pi''(h, s)$.
    This is clearly upper bounded by 
    \[
        \max_{a' \in \cA} \left(r_2(s, a') + \bE_{s' \sim \trans(s, a')}[v_2^{\pi_2}(h + (s, a'), s')]\right),
    \]
    since the latter is obtained by taking the maximum over $a'$.
    This finishes the proof.
\end{proof}

\begin{proof}[Proof of Lemma~\ref{lem:convex}]
    We only need to prove $\fr_{s, a}$ is convex.
    Consider any two points $(x_1, y_1)$ and $(x_2, y_2)$ in the feasible region $\fr_{s, a}$, and feasible-after-$s$ strategies $\Pi_1$ and $\Pi_2$ that induced these points.
    Without loss of generality, suppose $\Pi_1 = \Pi_1 \mid (s, a)$ (otherwise let $\Pi_1 \gets \Pi_1 \mid (s, a)$) and $\Pi_2 = \Pi_2 \mid (s, a)$.
    For any $\alpha \in (0, 1)$, the strategy $\Pi = \alpha \cdot \Pi_1 + (1 - \alpha) \cdot \Pi_2$ obtained by running $\Pi_1$ with probability $\alpha$ and $\Pi_2$ with probability $1 - \alpha$ gives utilities $\alpha \cdot (x_1, y_1) + (1 - \alpha) \cdot (x_2, y_2)$ in state $s$ after playing action $a$.
    We only need to argue that $\Pi = \alpha \cdot \Pi_1 + (1 - \alpha) \cdot \Pi_2$ is feasible after $s$.

    Consider any $h \in \cH^\Pi = \cH^{\Pi_1} \cup \cH^{\Pi_2}$, $s' > s$ where $\ap(s') = 2$, and $a' \in \cA$ such that $h + (s', a') \in \cH^\Pi = \cH^{\Pi_1} \cup \cH^{\Pi_2}$.
    We only need to show
    \[
        v_2^{\Pi \mid (h + (s', a'))}(h, s') \ge u^p(s').
    \]
    Observe that there is some $\beta \in [0, 1]$ such that
    \[
        \Pi \mid (h + (s', a')) = \beta \cdot (\Pi_1 \mid (h + (s', a'))) + (1 - \beta) \cdot (\Pi_2 \mid (h + (s', a'))),
    \]
    where $\beta$ is not necessarily equal to $\alpha$ due to conditioning.
    This means
    \[
        v_2^{\Pi \mid (h + (s', a'))}(h, s') = \beta \cdot v_2^{\Pi_1 \mid (h + (s', a'))}(h, s') + (1 - \beta) \cdot v_2^{\Pi_2 \mid (h + (s', a'))}(h, s').
    \]
    Since $\Pi_1$ and $\Pi_2$ are both feasible after $s$, we have
    \[
        \beta \cdot v_2^{\Pi_1 \mid (h + (s', a'))}(h, s') + (1 - \beta) \cdot v_2^{\Pi_1 \mid (h + (s', a'))}(h, s') \ge \beta \cdot u^p(s') + (1 - \beta) \cdot u^p(s') = u^p(s').
    \]
    This finishes the proof.
\end{proof}

\begin{proof}[Proof of Lemma~\ref{lem:evaluation}]
    We first prove $\bE_{s' \sim \trans(s, a)}\left[p_{s'}\right] \in \pf_{s, a}$, and
    \[
        \alpha \cdot \pf_{s, a}(\alpha) \ge \alpha \cdot \left((r_1(s, a), r_2(s, a)) + \bE_{s' \sim \trans(s, a)}\left[p_{s'}\right]\right).
    \]
    Let $\Pi_{s'}$ and $a_{s'}$ be the strategy and action corresponding to $p_{s'}$ for each $s' > s$, respectively.
    Clearly $\Pi_{s'}$ is feasible after $s'$, and by construction the constraint in $s'$ (if there is one) is also satisfied by $\Pi_{s'}$.
    Moreover, for each $i \in \{1, 2\}$,
    \[
        u_i^{\Pi_{s'} \mid (s', a_{s'})}(\emptyset, s') = (p_{s'})_i.
    \]
    Now we only need to construct a strategy $\Pi$ satisfying: (1) $\Pi$ is feasible after $s$, and (2) for each $i \in \{1, 2\}$,
    \[
        u_i^{\Pi \mid (s, a)}(\emptyset, s) = \bE_{s' \sim \trans(s, a)}\left[(p_{s'})_i\right].
    \]
    This is achieved by the following construction: Draw $\pi_{s'} \sim \Pi_{s'} \mid (s', a_{s'})$ for each $s' > s$.
    Let $\pi(\emptyset, s) = a$.
    For each $(h, s'')$ where $h = (s_1, a_1, \dots, s_t, a_t)$, let
    \[
        \pi(h, s) = \pi_{s_1}((s_2, a_2, \dots, s_t, a_t), s'').
    \]
    Let $\Pi$ be the distribution of $\pi$.
    One can check $\Pi$ satisfies the two desired conditions.

    Now we prove
    \[
        \alpha \cdot \pf_{s, a}(\alpha) \le \alpha \cdot \left((r_1(s, a), r_2(s, a)) + \bE_{s' \sim \trans(s, a)}\left[p_{s'}\right]\right).
    \]
    Suppose otherwise.
    Let $\Pi$ be a strategy that is feasible after $s$ that implements $\pf_{s, a}(\alpha)$.
    There must be some $s' > s$ such that
    \[
        \alpha \cdot (v_1^{\Pi \mid (s, a)}(\emptyset, s'), v_2^{\Pi \mid (s, a)}(\emptyset, s')) > \alpha \cdot p_{s'}.
    \]
    This means there exists a strategy $\Pi_{s'} = \Pi \mid (s, a)$ such that
    \[
        \alpha \cdot \bE_{\pi_{s'} \sim \Pi_{s'}}\left[(v_1^{\Pi_{s'} \mid (s', \pi_{s'}(\emptyset, s'))}(\emptyset, s'), v_2^{\Pi_{s'} \mid (s', \pi_{s'}(\emptyset, s'))}(\emptyset, s'))\right] > \alpha \cdot p_{s'},
    \]
    which means there exists some $a'$ such that
    \[
        \alpha \cdot (v_1^{\Pi_{s'} \mid (s', a')}(\emptyset, s'), v_2^{\Pi_{s'} \mid (s', a')}(\emptyset, s')) > \alpha \cdot p_{s', a'}.
    \]
    This contradicts the definition of $p_{s', a'}$, since $\Pi_{s'}$ is feasible after $s$.
\end{proof}

\begin{proof}[Proof of Lemma~\ref{lem:resolution}]
    We prove the claim inductively.
    Consider $s = n - 1$ first.
    For each $a \in \cA$, $\pf_{s, a}$ consists of a single point $(r_1(s, a), r_2(s, a)$, so the claim holds trivially.

    Now fix some $s \in [n - 1]$ and suppose the claim holds for all $s' > s$.
    Fix an action $a \in \cA$ and consider the first bullet point.
    By Lemma~\ref{lem:evaluation}, for any $\alpha \in \bR_+^2$,
    \[
        \pf_{s, a}(\alpha) = (r_1(s, a) + r_2(s, a)) + \bE_{s' \sim \trans(s, a)}[p_{s'}] = (r_1(s, a) + r_2(s, a)) + \sum_{s' > s} \trans(s, a, s') \cdot p_{s'}.
    \]
    Here, for each $s' > s$, $(p_{s'})_2$ is a multiple of $2^{-(n - s - 1)L}$, and $(p_{s'})_1$ is a multiple of $C_{s + 1}$ because of the induction hypothesis.
    Since $r_1(s, a)$, $r_2(s, a)$ and $\trans(s, a, s')$ are multiples of $2^{-L}$, $(\pf_{s, a}(\alpha))_2$ must be a multiple of $2^{-(n - s)L}$, and $(\pf_{s, a}(\alpha))_1$ must be a multiple of $1 / (2^L C_{s + 1})$.

    Now if $\ap(s) = 2$, we need to further consider the second bullet point.
    When there is no point on $\pf_{s, a}$ whose $y$-axis is precisely $u^p(s)$, we know $\pp(s, a)$ is either $(-n, -n)$ or some turning point on $\pf_{s, a}$.
    In both cases, $(\pp(s, a))_2$ is a multiple of $2^{-(n - s)L}$, and $(\pp(s, a))_1$ is a multiple of $1 / (2^L C_{s + 1})$.
    Alternatively, when there is a point on $\pf_{s, a}$ whose $y$-axis is precisely $u^p(s)$, this point must be $\pp(s, a)$.
    Moreover, there exist two turning points $p_1$ and $p_2$ on $\pf_{s, a}$ such that $\pp(s, a)$ is the unique convex combination of $p_1$ and $p_2$ whose $y$-axis is $u^p(s)$.
    That is,
    \[
        \pp(s, a) = \left(\frac{(p_1)_2 - u^p(s)}{(p_1)_2 - (p_2)_2} \cdot (p_2)_1 + \frac{u^p(s) - (p_2)_2}{(p_1)_2 - (p_2)_2} \cdot (p_1)_1, u^p(s)\right).
    \]
    Here, $u^p(s)$ is a multiple of $2^{-(n - s)L}$ (we will prove this later), so $(\pp(s, a))_2)$ is a multiple of $2^{-(n - s)L}$.
    $(p_1)_1$ and $(p_2)_1$ are multiples of $1 / (2^L C_{s + 1})$.
    As for the coefficients of $(p_1)_1$ and $(p_2)_1$, observe that $(p_1)_2$ and $(p_2)_2$ are multiples of $2^{-(n - s)L}$, and they are between $0$ and $n - s$.
    So there must be some integer $k \le 2^{(n - s)L} \cdot (n - s)$ such that both coefficients are multiples of $1 / k$.
    As a result, $(\pp(s, a))_1$ is a multiple of $1 / (2^L C_{s + 1} k)$, and we can let $C_s = 2^L C_{s + 1} k$, which satisfies
    \[
        C_s \le 2^L \cdot 2^{(n - s - 1)(n + 1)L} \cdot 2^{(n - s)L} \cdot (n - s) \le 2^{(n - s - 1)(n + 1)L} \cdot 2^{(n + 1)L} \le 2^{(n - s)(n + 1)L}.
    \]

    Finally we quickly argue that when $\ap(s) = 2$, $u^p(s)$ is a multiple of $2^{-(n - s)L}$.
    Recall that $u^p(s)$ is player $2$'s maximum onward utility in the subgame induced by $s$ when the other player tries to minimize player $2$'s utility.
    Without loss of generality, the equilibrium strategy is deterministic and history-independent (such a strategy can be computed by backward induction, for example).
    Again we can bound player $2$'s onward utility inductively.
    In state $n - 1$, player $2$'s onward utility must be $r_2(n - 1, a)$ for some action $a$, which means it is a multiple of $2^{-L}$.
    In state $n - 2$, player $2$'s onward utility can be written as the sum of $r_2(n - 2, a)$ for some action $a$, and the product of $\trans(n - 2, a, n - 1)$ and player $2$'s utility in state $n - 1$.
    So this utility must be a multiple of $2^{-2L}$.
    Repeating this argument for each $s$, we can show that player $2$'s utility in each state $s$ is a multiple of $2^{(n - s)L}$.
    This concludes the proof.
\end{proof}

\begin{proof}[Proof of Lemma~\ref{lem:binary_search}]
    Without loss of generality suppose $\ell_1 \le \frac12$ (otherwise we can flip the two axes and apply the same argument).
    Since $\|\ell - r\|_1 < 1 / (3n \cdot 2^{2n^2 L})$, we also have $r_1 \le \frac23$.
    Suppose otherwise, i.e., there is another point $q \in \pf_{s, a}$ between $q_\ell$ and $q_r$, such that $q$ is strictly above the line defined by $q_\ell$ and $q_r$, i.e.,
    \[
        \frac{q_2 - (q_\ell)_2}{q_1 - (q_\ell)_1} > \frac{(q_r)_2 - q_2}{(q_r)_1 - q_1}.
    \]
    Recall that $\max\{|q_\ell|, |q_r|, |q|\} \le n$.
    Moreover, by Corollary~\ref{cor:resolution}, all quantities in the above inequality are multiples of $1/C$, where $C \le 2^{n^2 L}$.
    Observe that
    \[
        \frac{q_2 - (q_\ell)_2}{q_1 - (q_\ell)_1} - \frac{(q_r)_2 - q_2}{(q_r)_1 - q_1} = \frac{C^2 ((q_2 - (q_\ell)_2)((q_r)_1 - q_1) - ((q_r)_2 - q_2)(q_1 - (q_\ell)_1))}{C^2 (q_1 - (q_\ell)_1)((q_r)_1 - q_1)}.
    \]
    Here, both the numerator and the denominator are integers, and the denominator is no larger than $C^2 \cdot n \le n \cdot 2^{2n^2 L}$.
    Since the fraction is strictly positive, we must have
    \[
        \frac{q_2 - (q_\ell)_2}{q_1 - (q_\ell)_1} - \frac{(q_r)_2 - q_2}{(q_r)_1 - q_1} \ge \frac{1}{n \cdot 2^{2n^2 L}}.
    \]

    On the other hand, since $q_\ell = \pf_{s, a}(\ell)$ and $q_r = \pf_{s, a}(r)$, we have
    \[
        -\frac{\ell_1}{\ell_2} \ge \frac{q_2 - (q_\ell)_2}{q_1 - (q_\ell)_1} \quad \text{and} \quad -\frac{r_1}{r_2} \le \frac{(q_r)_2 - q_2}{(q_r)_1 - q_1}.
    \]
    This implies
    \[
        \frac{r_1}{r_2} - \frac{\ell_1}{\ell_2} \ge \frac{1}{n \cdot 2^{2n^2 L}}.
    \]
    So we have
    \begin{align*}
        \|\ell - r\|_1 & = 2(r_1 - \ell_1) \ge 2 r_2 \cdot \frac{r_1 - \ell_1}{r_2} \ge 2 r_2 \cdot \frac{r_1}{r_2} - 2 \ell_2 \cdot \frac{\ell_1}{\ell_2} \\
        & = 2 \ell_2 \cdot \frac{r_1}{r_2} + 2 (r_2 - \ell_2) \cdot \frac{r_1}{r_2} - 2 \ell_2 \cdot \frac{\ell_1}{\ell_2} \\
        & \ge \frac{2 \ell_2}{n \cdot 2^{2n^2 L}} - \|\ell - r\|_1 \cdot \frac{r_1}{r_2}.
    \end{align*}
    Recall that without loss of generality, $\ell_1 \le \frac12$ (so $\ell_2 \ge \frac12$) and $r_1 \le \frac23$ (so $r_2 \ge \frac13$).
    Plugging these in and rearranging terms, the above inequality implies 
    \[
        3 \|\ell - r\|_1 \ge \|\ell - r\|_1 + \frac{r_1}{r_2} \cdot \|\ell - r\|_1 \ge \frac{2 \ell_2}{n \cdot 2^{2n^2 L}} \ge \frac{1}{n \cdot 2^{2n^2 L}} \implies \|\ell - r\|_1 \ge \frac{1}{3n \cdot 2^{2n^2 L}},
    \]
    a contradiction.
\end{proof}

\section{Omitted Algorithms and Proofs in Section~\ref{sec:efce}}
\label{app:efce}

\begin{proof}[Proof of Lemma~\ref{lem:approximate_binary_search}]
    Without loss of generality suppose $\ap(s) = 1$.
    We show this in two steps.
    First imagine the ``minimum'' curve possible given approximate evaluations satisfying the condition stated in the lemma.
    This is the curve that the binary search actually operates on in the worst case.
    As illustrated in the left subfigure of Figure~\ref{fig:approximate_binary_search}, this minimum curve is the blue one, which is lower than the actual (black) curve at most by $\frac{n - s - 1}{n} \cdot \eps$ in every direction.
    Also recall that the blue dashed line is obtained by shifting the black dashed line to the left by $\frac{n - s - 1}{n} \cdot \eps$.
    These facts imply that the $x$-coordinate of the blue hollow point is smaller than that of the pivotal point (the black hollow point) by at most $\frac{n - s - 1}{n} \cdot \eps$.

    Now consider how well the binary search approximates the blue hollow point.
    Suppose $\ell$, $r$, $q_\ell$ and $q_r$ are as illustrated in the right subfigure of Figure~\ref{fig:approximate_binary_search}.
    We need to bound the distance between the blue hollow point and the red one.
    Recall that $\|\ell - r\|_1 < \frac{\eps}{10 n^2}$, which means the angle $\theta$ between $\ell$ and $r$ is smaller than $\frac{\eps}{2 n^2}$.
    Moreover, observe that $\theta$ upper bounds the sum of the two acute angles in the triangle containing the red segment, so the angle at $q_\ell$ is at most $\theta \le \frac{\eps}{2 n^2}$.
    This implies that the distance between the two points we care about is at most the length of the segment to the left of $q_\ell$, times $\sin \theta$.
    The length of the segment is at most $\sqrt{2} \cdot n$, and $\sin \theta \le \theta \le \frac{\eps}{2 n^2}$, so the distance is at most $\eps / n$.
    Putting the two parts together, we conclude that $(p_{s, a})_{3 - \ap(s)} \ge \pp(s, a)_{3 - \ap(s)} - \frac{n - s}{n}(s)$.
\end{proof}

\begin{proof}[Proof of Lemma~\ref{lem:key_efce}]
    Apply induction on $s = n - 1, \dots, 1$.
    When $s = n - 1$, the first bullet point holds because $\eval(s, a, \alpha)$ is alwaus exact.
    Lemma~\ref{lem:approximate_binary_search} then implies the second bullet point.
    Now fix some $s$ and suppose the two bullet points hold for all $s' > s$.
    Consider lines 3-12 in Algorithm~\ref{alg:approximate_eval}, where $\cD(s, a, \alpha)$ is recursively computed.
    Let $q^*_{s', a'} = \argmax_{q \in \pf_{s', a'}: q_{\ap(s')} \ge u^p(s')} \alpha \cdot q$, and $q^*_{s'} = \argmax_{q \in \{q^*_{s', a'}\}_{a'}} \alpha \cdot q$.
    By the induction hypothesis, we have
    \[
        \alpha \cdot q_{s', a'} \ge \alpha \cdot q^*_{s', a'} - \frac{n - s'}{n} \cdot \eps \ge \alpha \cdot q^*_{s', a'} - \frac{n - s - 1}{n} \cdot \eps.
    \]
    As a result, we have $\alpha \cdot q_{s'} \ge \alpha \cdot q^*_{s'} - \frac{n - s - 1}{n} \cdot \eps$, and therefore $\alpha \cdot \cD(s, a, \alpha) \ge \alpha \cdot \pf_{s, a}(\alpha) - \frac{n - s - 1}{n} \cdot \eps$.
    Lemma~\ref{lem:approximate_binary_search} then implies the second bullet point.
\end{proof}

\begin{algorithm}[!ht]
\KwIn{a turn-taking stochastic game $(\cS = [n], \cA, \ap, r_1, r_2, \trans)$, an objective direction $\alpha_\obj$ where $\|\alpha_\obj\|_1 \le 1$, a desired accuracy $\eps$.}
\KwOut{an approximately optimal objective value under EFCE, together with an implicit representation of an $\eps$-EFCE achieving the approximate objective value in the input game.}
    create a data structure $\cD$ that stores the results of all evaluations (used by $\eval$)\;
    for each state $s \in \cS$, let $\hat{u^p}(s) \gets u^p(s) - \frac{n - s - 1}{n} \cdot \eps$\;
    \For{each state $s = n - 1, n - 2, \dots, 1$}{
        \For{each action $a \in \cA$}{
            let $(\ell, r) \gets ((1, 0), \alpha_\obj)$ if $\ap(s) = 1$, and $(\ell, r) \gets ((0, 1), \alpha_\obj)$ if $\ap(s) = 2$\;
            let $q_\ell \gets \eval(s, a, \ell)$, $q_r \gets \eval(s, a, r)$\;
            if $(q_\ell)_{\ap(s)} < \hat{u^p}(s)$, let $p_{s, a} \gets (-n, -n)$\;
            if $(q_r)_{\ap(s)} \ge \hat{u^p}(s)$, let $p_{s, a} \gets q_r$\;
            \If{$(q_\ell)_{\ap(s)} \ge \hat{u^p}(s)$ and $(q_r)_{\ap(s)} < \hat{u^p}(s)$}{
                \While{$\|\ell - r\|_1 \ge \frac{\eps}{10n^2}$}{
                    let $q \gets \eval(s, a, (\ell + r) / 2)$ (see Algorithm~\ref{alg:approximate_eval})\;
                    let $\ell \gets (\ell + r) / 2$ if $q_{\ap(s)} \ge \hat{u^p}(s)$, and $r \gets (\ell + r) / 2$ otherwise\;
                }
                let $q_\ell \gets \eval(s, a, \ell)$, $q_r \gets \eval(s, a, r)$, $\ell_{s, a} \gets \ell$, $r_{s, a} \gets r$\;
                let $p_{s, a} \gets \frac{(q_\ell)_{\ap(s)} - \hat{u^p}(s)}{(q_\ell)_{\ap(s)} - (q_r)_{\ap(s)}} \cdot q_r + \frac{\hat{u^p}(s) - (q_r)_{\ap(s)}}{(q_\ell)_{\ap(s)} - (q_r)_{\ap(s)}} \cdot q_\ell$\;
            }
        }
    }
    let $\opt \gets \max_{a \in \cA}\ (p_{\si, a})_1$\;
    \Return $\opt$, $\{p_{s, a}\}_{s, a}$ $\{\ell_{s, a}\}_{s, a}$, $\{r_{s, a}\}_{s, a}$, and $\cD$\;
\caption{An algorithm for computing an approximately optimal $\eps$-EFCE in turn-taking stochastic games, in time $\mathrm{poly}(n, m, \log(1 / \eps))$.}
\label{alg:efce}
\end{algorithm}

\begin{algorithm}[!ht]
\KwIn{a state $s$, an action $a$, a direction of evaluation $\alpha$, all variables in Algorthm~\ref{alg:efce}.}
\KwOut{an approximation of $\pf_{s, a}(\alpha)$.}
    if $s = \st = n$ then \Return $(0, 0)$\;
    \If{$(s, a, \alpha) \notin \cD$ (i.e., if $\cD(s, a, \alpha)$ does not exist)}{
        \For{$s' = s + 1, \dots, n$}{
            \For{$a' \in \cA$}{
                let $q_{s', a'} \gets \eval(s', a', \alpha)$\;
                \If{$(q_{s', a'})_{\ap(s')} < (p_{s', a'})_{\ap(s')}$}{
                    let $q_{s', a'} \gets p_{s', a'}$\;
                }
            }
            let $q_{s'} \gets \argmax_{q \in \{q_{s', a'}\}_{a' \in \cA}} \alpha \cdot q$\;
        }
        let $\cD(s, a, \alpha) \gets (r_1(s, a), r_2(s, a)) + \bE_{s' \sim \trans(s, a)}[q_{s'}]$\;
    }
    \Return $\cD(s, a, \alpha)$\;
\caption{$\eval$: A subroutine of Algorithm~\ref{alg:efce} that performs approximate recursive evaluations as needed.}
\label{alg:approximate_eval}
\end{algorithm}

\begin{algorithm}[!ht]
\KwIn{A turn-taking stochastic game $(\cS, \cA, \ap, r_1, r_2, \trans)$, an objective direction $\alpha_\obj$, the output of Algorithm~\ref{alg:efce}, a history $h = (s_1, a_1, \dots, s_t, a_t)$, and a state $s$.}
\KwOut{$\pi(h, s)$, where $\pi \sim \Pi \mid h$, and $\Pi$ is the strategy encoded in the output of Algorithm~\ref{alg:efce}; $\error$ if $h$ is not an admissible history under $\Pi$.}
    let $\alpha \gets \alpha_\obj$, $a \gets \argmax_{a' \in \cA} \alpha \cdot p_{s, a'}$, $q \gets p_{s, a}$\;
    if $|h| = 0$ \Return $a$\;
    if $a_1 \ne a$, \Return $\error$\;
    \For{$i = 1, 2, \dots, t - 1$}{
        \If{$q = p_{s_i, a_i}$}{
            let $\ell \gets \ell_{s_i, a_i}$, $r \gets r_{s_i, a_i}$, $a_\ell \gets \argmax_{a' \in \cA} \ell \cdot \max^{\ap(s_{i + 1})}\{p_{s_{i + 1}, a'}, \cD(s_{i + 1}, a', \ell)\}$, $a_r \gets \argmax_{a' \in \cA} r \cdot \max^{\ap(s_{i + 1})}\{p_{s_{i + 1}, a'}, \cD(s_{i + 1}, a', \ell)\}$\;
            \tcc{for two points $q_1$ and $q_2$, $\max^k\{q_1, q_2\}$ denotes the point with the larger $k$-th coordinate between the two}
            let $\alpha \gets \ell$ if $a_{i + 1} = a_\ell$\;
            let $\alpha \gets r$ if $a_{i + 1} = a_r$\;
            if $a_{i + 1} \notin \{a_\ell, a_r\}$, \Return $\error$\;
        }
        \Else{
            let $a \gets \argmax_{a' \in \cA} \alpha \cdot \max^{\ap(s_{i + 1})}\{p_{s_{i + 1}, a'}, \cD(s_{i + 1}, a', \alpha)\}$\;
            if $a_{i + 1} \ne a$ \Return $\error$\;
        }
        let $q \gets \max^{\ap(s_{i + 1})}\{p_{s_{i + 1}, a_{i + 1}}, \cD(s_{i + 1}, a_{i + 1}, \alpha)\}$\;
    }
    \If{$q = p_{s_t, a_t}$}{
        let $\ell \gets \ell_{s_i, a_i}$, $a_\ell \gets \argmax_{a' \in \cA} \ell \cdot \max^{\ap(s)}\{p_{s, a'}, \cD(s, a', \ell)\}$, $q_\ell \gets \max^{\ap(s)}\{p_{s, a_\ell}, \cD(s, a_\ell, \ell)\}$;
        let $r \gets r_{s_i, a_i}$, $a_r \gets \argmax_{a' \in \cA} r \cdot \max^{\ap(s)}\{p_{s, a'}, \cD(s, a', r)\}$, $q_r \gets \max^{\ap(s)}\{p_{s, a_r}, \cD(s, a_r, r)\}$\;
        let $a \gets a_\ell$ with probability $\frac{q_r - q}{q_r - q_\ell}$, $a \gets a_r$ with probability $\frac{q - q_\ell}{q_r - q_\ell}$\;
    }
    \Else{
        let $a \gets \argmax_{a' \in \cA} \alpha \cdot \max^{\ap(s)}\{p_{s, a'}, \cD(s, a', \alpha)\}$\;
    }
    \Return $a$\;
\caption{A procedure that decodes the output of Algorithm~\ref{alg:efce}.}
\label{alg:decode_efce}
\end{algorithm}